\documentclass[submission,copyright,creativecommons]{eptcs}
\providecommand{\event}{FROM 2024} 

\usepackage{iftex}
\usepackage{graphicx,color}
\usepackage{amsmath,amssymb,amsfonts}

\ifpdf
  \usepackage{underscore}         
  \usepackage[T1]{fontenc}        
\else
  \usepackage{breakurl}           
\fi

\newcommand{\Aclass}{\Arg_{class}}
\newcommand{\Fclass}{\F_{class}}

\newcommand{\F}{\mathcal{O}t}
\newcommand{\Arg}{\mathcal{A}\!t}
\newcommand{\f}{f}

\newcommand{\C}{\xdom{C}}
\newcommand{\Conts}{\xdom{Cont}_S}
\newcommand{\Conta}{\xdom{Cont}_A}
\newcommand{\Cd}{\xdom{C}^{\mathcal{D}}}
\newcommand{\Cdc}{\ovr{\xdom{C}}^{\mathcal{D}}}
\newcommand{\Contd}{\Cont^{\mathcal{D}}}

\newcommand{\Ct}{\tilde{\xdom{C}}}

\newcommand{\Rd}{\xdom{R}}

\newcommand{\Dt}{\tilde{\xdom{D}}}
\newcommand{\Md}{\xdom{M}}

\newcommand{\D}{\xdom{D}}
\newcommand{\Od}{\xdom{O}}
\newcommand{\K}{\xdom{K}}
\newcommand{\inv}[4]{{#1}{\langle#2\rangle}\fun{1}{#3}{\langle#4\rangle}}
\newcommand{\ainv}[3]{#1{\langle#2,#3\rangle}}

\newcommand{\co}[2]{co(#1|#2)}

\newcommand{\Names}{\mathcal{N}}
\newcommand{\Conames}{\ovr{\mathcal{N}}}

\newcommand{\xrestr}[2]{#1\backslash#2}

\newcommand{\shole}{\circ}

\newcommand{\ahole}{\bullet}
\newcommand{\ars}{(;\,\bullet)}
\newcommand{\als}[1]{(#1\,;)}
\newcommand{\anu}[2]{(\xrestr{#1}{#2})}
\newcommand{\alp}[1]{(#1\parallel)}
\newcommand{\arp}[1]{(\parallel#1)}
\newcommand{\amatch}{match_{\alpha}}
\newcommand{\amatchn}{match_{\alpha}^{\Names}}
\newcommand{\inalpha}{in_{\alpha}}
\newcommand{\iotatwo}{\iota_2}
\newcommand{\glb}{g}
\newcommand{\pirestr}[2]{#1|_{#2}}
\newcommand{\pirestrn}[2]{#1|_{#2}^{\Names}}

\newcommand{\match}{match}
\newcommand{\perm}{perm}

\newcommand{\at}[2]{#1@#2}
\newcommand{\asyn}[2]{#1\,\cdot\,#2}
\newcommand{\xlbd}[2]{\lambda{#1}\,.\,#2}
\newcommand{\cres}{c_{\Res}}
\newcommand{\cb}{c_{\mu}}
\newcommand{\ca}{c_{\alpha}}

\newcommand{\ks}{ks}
\newcommand{\kd}{kd}

\newcommand{\restrr}{\,\,\backslash^{\!\!\Rc}\,\,}
\newcommand{\seqr}{\,\xseq^{\!\!\Rc}\,}
\newcommand{\parr}{\,\xmerge^{\!\!\Rc}\,}

\newcommand{\restrden}{\,\,\widehat{\backslash}\,\,}
\newcommand{\seqden}{\,\,\widehat{\xseq}\,\,}
\newcommand{\parden}{\,\,\widehat{\xmerge}\,\,}

\newcommand{\xinfer}{\,\nearrow\,}
\newcommand{\xstep}[3]{#1\,\stackrel{#2}{\longrightarrow}\,#3}
\newcommand{\narrow}{\,{\longrightarrow\!\!\!\!\!\!\!\!\!\!/}\,\,\,\,}

\newcommand{\idbag}{id}
\newcommand{\applybag}[2]{#1(#2)}
\newcommand{\subsbag}[3]{\langle#1\mid#2\mapsto#3\rangle}

\newcommand{\s}{x}
\newcommand{\Stmt}{Stmt}
\newcommand{\GStmt}{GStmt}
\newcommand{\cx}{wgt}
\newcommand{\xop}{\,\,\mathsf{op}\,\,}

\newcommand{\U}{U}

\newcommand{\Res}{Res}
\newcommand{\KRes}{KRes}
\newcommand{\SRes}{SRes}
\newcommand{\ARes}{ARes}
\newcommand{\Rc}{R}
\newcommand{\Ids}{Ids}
\newcommand{\E}{E}
\newcommand{\xcard}[1]{|#1|}

\newcommand{\xlen}{len}
\newcommand{\phie}{\phi_E}

\newcommand{\oprestr}{\backslash}
\newcommand{\opcnu}{\tilde{\backslash}}
\newcommand{\opcseq}{add_;}
\newcommand{\opclpar}{add_{\xlmerge}}
\newcommand{\opclsyn}{add_{\xlsyn}}

\newcommand{\cnu}[2]{#1\,\tilde{\backslash}\,#2}
\newcommand{\cseq}[2]{add_{;}(#1,#2)}
\newcommand{\clpar}[2]{add_{\xlmerge}(#1,#2)}
\newcommand{\clsyn}[2]{add_{\xlsyn}(#1,#2)}

\newcommand{\gammaz}{\gamma_0}

\newcommand{\Semd}{F_D}

\newcommand{\xcardc}{card_{\gamma}}
\newcommand{\xcardu}{card_u}
\newcommand{\xwhere}{\,\,\,\mathsf{where}\,\,\,}
\newcommand{\xif}{\,\,\mathsf{if}\,\,}
\newcommand{\xthen}{\,\,\mathsf{then}\,\,}
\newcommand{\xelse}{\,\,\mathsf{else}\,\,}

\newcommand{\xotherwise}{\,\,\mathsf{otherwise}\,\,}
\newcommand{\xlet}{\,\,\mathsf{let}\,\,}
\newcommand{\xin}{\,\,\mathsf{in}\,\,}

\newcommand{\dsem}{{\mathcal{D}}}
\newcommand{\osem}{{\mathcal{O}}}

\newcommand{\msync}{msync}

\newcommand{\Semx}{\xi_P}
\newcommand{\semx}{\xi_Q}

\newcommand{\semc}[1]{[\![\hspace*{0.02cm}#1\hspace*{0.02cm}]\!]_C}
\newcommand{\semk}[1]{[\![\hspace*{0.02cm}#1\hspace*{0.02cm}]\!]_K}
\newcommand{\semf}[1]{[\![\hspace*{0.02cm}#1\hspace*{0.02cm}]\!]_F}

\newcommand{\msnd}[1]{[#1]^{snd}}
\newcommand{\mrcv}[1]{[#1]^{rcv}}

\newcommand{\join}{\,\&\,}

\newcommand{\Jn}{J^n}
\newcommand{\Jpn}{J_P^n}

\newcommand{\ccsn}{CCS^n}
\newcommand{\ccsnp}{CCS^{n+}}

\newcommand{\Lccsn}{{\mathcal{L}}_{CCS^n}}
\newcommand{\Lccsnp}{{\mathcal{L}}_{CCS^{n+}}}

\newcommand{\df}{d_{F}}
\newcommand{\dc}{d_C}

\newcommand{\xundef}{\uparrow}

\newcommand{\Den}{\xdom{Den}}
\newcommand{\Sem}{\xdom{Sem}}

\newcommand{\bag}[1]{\langle\!\!|#1|\!\!\rangle}

\newcommand{\Dd}{\D^{\dsem}}

\newcommand{\opa}{op_A}

\newcommand{\Lang}{L}

\newcommand{\X}{\xdom{X}}
\newcommand{\Nset}{\mathbb{N}}

\newcommand{\Rset}{\mathbb{R}}
\newcommand{\mset}[1]{[#1]}
\newcommand{\notation}{\ \stackrel{\mathrm{not.}}{=}\ }
\newcommand{\xbnf}{::=}

\newcommand{\Lsyn}{{\mathcal{L}}_{\mathit{syn}}}

\newcommand{\Act}{Act}
\newcommand{\IAct}{IAct}
\newcommand{\SAct}{SAct}

\newcommand{\pnco}{\pset_{\mathit{nco}}}
\newcommand{\pco}{\pset_{\mathit{co}}}
\newcommand{\pfin}{\pset_{\mathit{fin}}}

\newcommand{\pset}{{\mathcal{P}}}
\newcommand{\mc}[1]{{\mathcal{#1}}}

\newcommand{\Id}{Id}
\newcommand{\var}[3]{\langle #1 \mid #2 \mapsto #3 \rangle}

\newcommand{\xapply}[2]{#1(#2)}

\newcommand{\fdom}[1]{\,{\mathsf{dom}}\hspace*{0.05cm}(#1)\,}

\newcommand{\xnedp}[1]{\hspace*{0.02cm}\oplus^{{\hspace*{-0.03cm}{#1}}}\hspace*{0.02cm}}
\newcommand{\xned}{\oplus}

\newcommand{\Cont}{\xdom{Cont}}

\newcommand{\Kres}{Kres}

\newcommand{\Conf}{Conf}
\newcommand{\Qd}{\xdom{Q}_D}
\newcommand{\Qo}{\xdom{Q}_O}
\newcommand{\Pd}{\xdom{P}_D}
\newcommand{\Po}{\xdom{P}_O}

\newcommand{\frestr}[2]{#1\!\!\upharpoonright\!#2}

\newcommand{\Semo}{Sem_O}
\newcommand{\scontext}{S}

\newcommand{\fix}{\mathsf{fix}}

\newcommand{\Decl}{Decl}

\newcommand{\A}{A}

\newcommand{\Y}{Y}

\newcommand{\fun}[1]{\mathop{\stackrel{#1}{\rightarrow}}}
\newcommand{\half}{\frac{1}{2}}

\newcommand{\xstop}{\,\mathsf{stop}\,}

\newcommand{\xseq}{\hspace*{0.025cm};\hspace*{0.04cm}}
\newcommand{\xmerge}{\hspace*{0.025cm}\parallel\hspace*{0.03cm}}
\newcommand{\xlmerge}{\hspace*{0.02cm}\lfloor\hspace*{-0.1cm}\lfloor\hspace*{0.055cm}}
\newcommand{\xsyn}{\hspace*{0.055cm}|\hspace*{0.055cm}}
\newcommand{\xlsyn}{\hspace*{0.02cm}\lfloor\hspace*{0.055cm}}
\newcommand{\xsep}{\hspace*{0.035cm}\,\big{|}\,\hspace*{0.035cm}}

\newcommand{\ovr}[1]{\overline{#1}}

\newcommand{\den}[2]{{\mathcal{#1}}[\![#2]\!]}
\newcommand{\sem}[1]{[\![\hspace*{0.02cm}#1\hspace*{0.02cm}]\!]}

\newcommand{\az}{\alpha_0}
\newcommand{\kz}{k_0}

\newcommand{\nmax}{\overline{n}}

\newcommand{\Bool}{Bool}
\newcommand{\true}{{\mathsf{true}}}
\newcommand{\false}{{\mathsf{false}}}

\newcommand{\xdom}[1]{{\mathbf{#1}}}

\newcommand{\qed}{\,\Box}

\newtheorem{thm}{Theorem}
\newtheorem{cor}{Corollary}
\newtheorem{lem}{Lemma}

\newtheorem{prop}{Proposition}

\newtheorem{defn}{Definition}
\newtheorem{exmp}{Example}

\newtheorem{rem}{Remark}

\newtheorem{notatie}[thm]{Notation}
\newenvironment{proof}{\noindent{\em Proof.}}{\hspace*{1mm} \hfill$\Box$ \par \medskip }
\newenvironment{proof*}[1]{\noindent{\bf{Proof of {#1}}}}{$\qed$\\}

\title{Abstract Continuation Semantics for \\
Multiparty Interactions in Process Calculi based on CCS}
\author{Eneia Nicolae Todoran 
\institute{Department of Computer Science\\
Technical University of Cluj-Napoca\\
Cluj-Napoca, Romania}
\email{eneia.todoran@cs.utcluj.ro}
\and
Gabriel Ciobanu
\institute{Academia Europaea\\
 London, United Kindom\\
 https://www.ae-info.org}
\email{gabriel@info.uaic.ro}
}
\def\titlerunning{Continuation Semantics for 
Multiparty Interactions in Process Calculi based on CCS}
\def\authorrunning{E.N. Todoran \& G. Ciobanu}
\begin{document}
\maketitle

\begin{abstract}
We develop denotational and operational semantics designed with continuations
for process calculi based on Milner’s CCS extended with mechanisms offering
support for multiparty interactions. We investigate the abstractness of this
continuation semantics. We show that our continuation-based denotational
models are weakly abstract with respect to the corresponding operational models.
\end{abstract}

\section{Introduction}

In denotational semantics, continuations have a long tradition, being used to
model a large variety of control mechanisms \cite{str00,dyb89,fri86}.
However, it is usually considered that continuations do not perform well enough
as a tool for describing concurrent behaviour~\cite{mos10}.
In \cite{ent00,ct14}, we introduced a technique for denotational and
operational semantic design named {\em{continuation semantics for concurrency
(CSC)}} which can be used to handle advanced concurrent control mechanisms
\cite{ct20,ct22,ent19}. The distinctive characteristic of the CSC technique
is the modelling of continuations as structured configurations of computations.

In this paper, we employ the CSC technique in providing denotational and 
operational semantics for the multiparty interaction mechanisms incorporated 
in two process calculi, namely $\ccsn$ and $\ccsnp$~\cite{lv10}, both based 
on the well-known CCS \cite{mil89} -- $\ccsn$ and $\ccsnp$ extend CCS with 
constructs called {\em{joint input}} and {\em{joint prefix}}, respectively, 
that can be used to express multiparty synchronous interactions. The semantic 
models are developed using the methodology of metric semantics \cite{bv96}.

In particular, we investigate the abstractness of continuation semantics.
As it is known, the completeness condition of the full abstraction criterion
\cite{mil77} is often difficult to be fulfilled. In models designed with
continuations, the problem may be even more difficult~\cite{car94,ct17}.
Therefore, in \cite{ct17,ent19} we introduced a {\em{weak abstractness}}
optimality criterion which preserves the correctness condition, but relaxes
the completeness condition of the classic full abstractness criterion. The
weak abstractness criterion comprises a weaker completeness condition called
{\em{weak completeness}}, which is easier to establish because it needs to be
checked only for denotable continuations (that handle only computations denotable
by the language constructs and represent an invariant of the computation).
We study the abstractness of continuation semantics based on the weak abstractness
criterion.  The continuation-based denotational models presented in this article
for the multiparty interaction mechanisms incorporated in $\ccsn$ and $\ccsnp$
are weakly abstract with respect to the corresponding operational models.

Following the approach presented in \cite{bv96}, we start from the language
$\Lsyn$ given in Chapter 11 of~\cite{bv96}. As it is mentioned in \cite{bv96},
the language $\Lsyn$ is ``essentially based on CCS''. Then, we consider two
language named $\Lccsn$ and $\Lccsnp$ which extend the language $\Lsyn$ with
constructs for multiparty interactions: $\Lccsn$ extends $\Lsyn$ with the
joint input construct of $\ccsn$, and $\Lccsnp$ extends $\Lsyn$ with the joint
prefix construct of $\ccsnp$. We define and relate continuation-based
denotational and operational semantics for $\Lccsn$ and $\Lccsnp$.

\noindent
\paragraph*{\bf Contribution:}
By using the methodology of metric semantics, we develop original continuation
semantics for the multiparty interaction mechanisms incorporated in $\ccsn$
and $\ccsnp$. We provide a new representation of continuations
based on a construction presented in Section \ref{theory.bags}.
We show that the denotational models presented in this paper are weakly
abstract with respect to the corresponding operational models. 
A weak abstractness result for $\ccsn$ was also presented in \cite{ct20};
the weak abstractness result for $\ccsnp$ is new.
The weak completeness condition of the weak abstractness principle presented
in \cite{ct17,ent19} should be checked only for denotable continuations.
Intuitively, the collection of denotable continuations have to be an invariant
of the computation, in the sense that it is sufficiently large to
support arbitrary computations denotable by program statements. The formal
conditions capturing this intuition are studied initially in \cite{ct17,ent19}.
In this article we offer a more general formal framework. We present the formal
conditions which guarantee that the domain of denotable continuations is
invariant under the operators used in the denotational semantics, where the
domain of denotable continuations is the metric completion of the class of
denotable continuations.

\section{Preliminaries} \label{theory}

We assume the reader is familiar with metric spaces, multisets, metric 
semantics \cite{bv96}, and the $\lambda$-calculus notation. For the used 
notions and notations, we refer the reader to~\cite{ct12,ct14,ct16,ct17}.

The notation $(x\in)X$ introduces the set $X$ with typical element $x$ 
ranging over~$X$. We write $S\subseteq{}X$ to express that $S$ is a subset of 
$X$. $\xcard{S}$ is the cardinal number of set~$S$. 
Let $X$ be a countable set.

The set of all finite multisets over $X$ is represented by using the notation 
$\mset{X}$; the construction~$\mset{X}$ and the operations on multisets that 
are specified formally in~\cite{ct16}. By a slight abuse, the cardinal number 
of a multiset $m\in\mset{X}$ defined as $\sum_{x\in\fdom{m}}m(x)$ is also 
denoted by~$\xcard{m}$. 
Even though the same notation $\xcard{\cdot}$ is used regardless 
of whether '$\cdot$' is a set or a multiset, it is always evident from 
the context whether the argument '$\cdot$' is a set or a multiset.
We represent a multiset by stringing its elements between square brackets 
'$[$' and '$]$'. For instance, the empty multiset is written as 
$\mset{}$, and $\mset{e_1,e_2,e_2}$ is the multiset with one and two occurrences 
of the elements $e_1$ and $e_2$, respectively. 
If $f\in{X}\to{Y}$ is a function (with domain $X$ and 
codomain $Y$) and $S$ is a subset of $X$, $S\subseteq{X}$, 
the notation $\frestr{f}{S}$ 
denotes the function $f$ restricted to the domain $S$, i.e. 
$\frestr{f}{S}:S\to{Y}$, $\frestr{f}{S}(x)=f(x),\forall{x\in{S}}$.
Also, if $f\in{X}\fun{}{Y}$ is a function, 
$\var{f}{x}{y}:X\fun{}Y$ is the function defined (for $x,x'{\in}X, y{\in}Y$) by:
\mbox{$\var{f}{x}{y}(x')=$} \mbox{$\!\!\!\xif x'{=}x \xthen y \xelse f(x')$}.
Given a function $f\in{X}\to{X}$, we say that an element $x\in{X}$
is a {\em{fixed point}} of $f$ if $f(x)=x$, and if this fixed point 
is unique we write $x=\fix(f)$. 

We present semantic models designed using the mathematical framework of 
{\em{1-bounded complete metric spaces}} \cite{bv96}. We assume the following 
notions are known: {\em{metric}} and {\em{ultrametric}} space, 
{\em{isometry}} (between metric spaces, denoted by '$\cong$'), {\em{Cauchy 
sequence}}, {\em{complete}} metric space, {\em{metric completion}}, 
{\em{compact}} set, and the {\em{discrete metric}}. 
We use the notion of {\em{metric domain}} as a synonym for the notion of 
complete (ultra) metric space. We assume the reader is familiar with the 
standard metrics for defining composed metric structures \cite{bv96}. We use 
the constructs for $\half$-identity, disjoint union ($+$), function space 
($\to$), Cartesian product ($\times$), and the compact powerdomain.
Every Cauchy sequence in a complete metric space~$M$ has a limit that is also 
in $M$. If $(M_1,d_1)$ and $(M_2,d_2)$ are metric spaces, a function 
$f{:}M_1\to{M_2}$ is a {\em{contraction}} if \mbox{$\exists{c\in\Rset}$, 
\mbox{$0\leq{c}<1$}, $\forall{x,y}\in{M_1}$}\, \mbox{$[d_2(f(x),f(x)){\leq}c\,
{\cdot}\,d_1(x,y)]$}. If \mbox{$c=1$} we say that $f$~is {\em{nonexpansive}}; 
each nonexpansive function is {\em{continuous}}~\cite{bv96}. The set of all 
nonexpansive functions from $M_1$ to $M_2$ is denoted by $M_1\fun{1}M_2$.
We recall Banach's theorem. 
\begin{thm}[Banach] \label{theory.2}
Let $(M,d)$ be a non-empty complete metric space. 
Each contraction  $f:M\to{M}$ has a {\em{unique}} fixed point.
\end{thm}

With $\pco(\cdot)$ ($\pnco(\cdot)$) we denote the power set of {\em{compact}} 
({\em{non-empty and compact}}) subsets of '$\cdot$'. $\pfin(\cdot)$ denotes 
the power set of {\em{finite}} subsets of '$\cdot$' (we always endow 
$\pfin(\cdot)$ with the discrete metric).

Hereafter, we shall often suppress the metrics part in metric domain definitions. 
For example, we write $\half\cdot{}M$ and $M_1\times{}M_2$ instead of 
$(M,d_{\half\cdot{}M})$ and $(M_1\times{}M_2,d_{M_1\times{}M_2})$, respectively.

Let $(M,d),(M',d')$ be metric spaces. We write $(M,d)\lhd(M',d')$, or simply 
$M\lhd{M'}$, to express that~$M$ is a {\em{subspace}} of $M'$, i.e, 
$M\subseteq{M'}$ and $\frestr{d'}{M}=d$ (the restriction of metric $d'$ to 
$M$ coincides with $d$).

For compact sets we use Theorem \ref{theory.metric.4} (due to Kuratowski) and
the characterization given in Theorem~\ref{theory.metric.5}. Given a complete
metric space $(M,d)$ and a subset $X$, $X\subseteq{}M$, according to Theorem
\ref{theory.metric.5}, the statement that $X$ is compact is equivalent to the
statement that $X$ is the limit (with respect to Hausdorff metric $d_H$) of a
sequence of finite sets \cite{bz83}. The proofs of these theorems are also
provided in~\cite{bv96}.

\begin{thm} \label{theory.metric.4} [Kuratowski] 
Let $(M,d)$ be a complete metric space.
\begin{itemize}
\item[(a)] If $(X_i)_i$ is a Cauchy sequence in $(\pnco(M),d_H)$ then
\item[]\quad$\lim_iX_i=\{ \lim_ix_i \mid \forall{}i:x_i\in{}X_i, (x_i)_i
\textnormal{\,is a Cauchy sequence in\,} M \}$.
\item[(b)] If $(X_i)_i$ is a Cauchy sequence in $(\pco(M),d_H)$ then either,
for almost all $i$, $X_i=\emptyset$, and $\lim_iX_i=\emptyset$, or for almost all
$i$ (say for $i\geq{}n$),
$X_i\neq\emptyset$ and
\item[]\quad$\lim_iX_i=\{ \lim_{i\geq{}n}x_i \mid
\forall{}i\geq{}n:x_i\in{}X_i, (x_i)_i \textnormal{\,is a Cauchy sequence in\,} M \}$.
\item[(c)] $(\pco(M),d_H)$ and $(\pnco(M),d_H)$ are complete metric spaces.
\end{itemize}
\end{thm}

\begin{thm} \label{theory.metric.5} Let $(M,d)$ be a complete metric space.
A subset $X\subseteq{}M$ is compact whenever $X=\lim_iX_i$, where each $X_i$
is a finite subset of $M$ (the limit is taken with respect to the Hausdorff metric $d_H$).
\end{thm}

\begin{rem} \label{jmath.0}
\begin{itemize}
\item[(a)] 
If $M$ and $M'$ are metric spaces with subspaces $S$ and $S'$ ($S\lhd{}M$ and 
$S'\lhd{}M'$), then $S+S'\lhd{}M+M'$, $S\times{}S'\lhd{}M\times{}M'$, 
$(\A\to{}S)\lhd(\A\to{}M)$, $\pco(S)\lhd\pco(M)$ and $\pnco(S)\lhd\pnco(M)$ 
(see \cite{bv96}, chapter~10).
\item[(b)]
Let $(M,d), (M_1,d_1)$ and $(M_2,d_2)$ be metric spaces. It is easy to verify that, 
if \mbox{$M_1\lhd{}M$}, \mbox{$M_2\lhd{}M$} and \mbox{$M_1\subseteq{}M_2$} then 
\mbox{$M_1\lhd{}M_2$}. 
\end{itemize}
\end{rem}

\begin{defn} \label{jmath.01}
Given a metric space $(M,d)$, a {\em{completion}} of $(M,d)$ is a complete 
metric space $(\ovr{M},d')$ such that $M\lhd{\ovr{M}}$ and for each element 
$x\in{\ovr{M}}$ we have: $x=\lim_jx_j$, with $x_j\in{M}, \forall{j}\in\Nset$ 
(limit is taken with respect to metric $d'$). 
\end{defn}

\noindent
Each metric space has a completion that is unique up to isometry \cite{bv96}.
For the proof of Remark \ref{jmath.03}, see~\cite{ct17}. 

\begin{rem} \label{jmath.03}
Let $(M,d)$ be a complete metric space, and $X$ be a subset of $M$, $X\subseteq{M}$.
We use the notation $\co{X}{M}$ for the set
$\co{X}{M}\notation\{ x \mid x\in{M}, x=\lim_ix_i, \forall{i\in\Nset:x_i\in{X}},
(x_i)_i$ is a Cauchy sequence in~$X\}$,
where limits are taken with respect to $d$ (as $(M,d)$ is complete 
$\lim_ix_i\in{M}$). If we endow~$X$ with $d_X=\frestr{d}{X}$ and $\co{X}{M}$ 
with $d_{\co{X}{M}}=\frestr{d}{\co{X}{M}}$, then $(\co{X}{M},d_{\co{X}{M}})$ is a 
metric completion of $(X,d_X)$. \ 
It is easy to see that $X\lhd\co{X}{M}$ and $\co{X}{M}\lhd{}M$. 
\end{rem}

\begin{rem} \label{jmath.04}
\begin{itemize}
\item[(a)]
If $(M_1,d_1)$ and $(M_2,d_2)$ are complete metric spaces and $(x_i)_i$ is a 
Cauchy sequence in $M_1+M_2$, then for almost all $i$ (i.e., for all but a 
finite number of exceptions) we have that $x_i=(1,x_i')$ or $x_i=(2,x_i')$, where 
$(x_i')_i$ is a Cauchy sequence in $M_1$ or $M_2$, respectively (see \cite{bv96}).
\item[(b)]
Let $(M_1,d_1)$ and $(M_2,d_2)$ be complete metric spaces. 
If $(x_1^i,x_2^i)_i$ is a Cauchy sequence in $M_1\times{}M_2$, then $(x_1^i)_i$ 
is a Cauchy sequence in $M_1$ and $(x_2^i)_i$ is a Cauchy sequence in $M_2$.
Since $M_1$ and $M_2$ are complete, there exists $x_1\in{}M_1$ and $x_2\in{}M_2$
such that $x_1=\lim_ix_1^i$ and $x_2=\lim_ix_2^i$, and 
$\lim_i(x_1^i,x_2^i)=(x_1,x_2)=(\lim_ix_1^i,\lim_ix_2^i)$ \cite{bv96}. 
\item[(c)] 
Let $(M,d)$ be a complete metric space. Let $(x_i)_i$ be a convergent sequence
in~$M$ with limit $x=\lim_ix_i$.
Then $(x_i)_i$ has a subsequence $(x_{f(i)})_i$ such that
\vspace*{-0.1cm}
\begin{equation}
\forall{}n\ \forall{}j\geq{}n\,[d(x_{f(j)},x)\leq2^{-n}] ,  \label{prop1}
\end{equation}
\vspace*{-0.1cm}
where $f\!\!:\!\Nset\!\to\!\Nset$ is a strictly monotone mapping, i.e.,
\mbox{$f(i)\!<\!f(i')\!$~whenever~$i\!<\!i'\!$}.
We obtain such a subsequence (by imposing the condition that the function
$f:\Nset\to\Nset$ is strictly monotone and) by putting $f(0)=0$, and if $i>0$
then $f(i)=m$, where $m\in\Nset$ is the smallest natural number such that
\vspace*{-0.3cm}
\begin{equation}
\forall{}l\geq{}m\,[d(x_l,x)\leq2^{-i}]. \label{prop2}
\end{equation}
\vspace*{-0.1cm}
It is easy to see that this subsequence satisfies property \eqref{prop1}.
Clearly, if $n=0$ (which implies $2^{-n}=1$), then \eqref{prop1} holds. 
If $n>0$ and $j\geq{}n$, then property \eqref{prop1} also holds
because we infer (from~\eqref{prop2}) that $d(x_{f(j)},x)\leq2^{-j}\leq2^{-n}$.
\end{itemize}
\end{rem}

\begin{lem} \label{theory.metric.13}
Let $(M,d)$, $(M_1,d_1)$ and $(M_2,d_2)$ be complete metric spaces.
Let $S, S_1$ and $S_2$ be subsets of $M$, $M_1$ and $M_2$, respectively,
$S\subseteq{}M$, $S_1\subseteq{}M_1$ and $S_2\subseteq{}M_2$.
Let $A$ be an arbitrary set. Then
\begin{itemize}
\item[(a)] $\co{S_1+{}S_2}{M_1+{}M_2}=\co{S_1}{M_1}+\co{S_2}{M_2}$,
\item[(b)] $\co{S_1\times{}S_2}{M_1\times{}M_2}=\co{S_1}{M_1}\times\co{S_2}{M_2}$,
\item[(c)] $\A\to\co{S}{M}=\co{\A\to{}S}{\A\to{}M}$,
\item[(d)] $\pco(\co{S}{M})=\co{\pco(S)}{\pco(M)}$,
\item[(e)] $\pnco(\co{S}{M})=\co{\pnco(S)}{\pnco(M)}$.
\end{itemize}
\end{lem}
\begin{proof}
Clearly, $S_1+S_2\subseteq{}M_1+M_2$, $S_1\times{}S_2\subseteq{}M_1\times{}M_2$,
$(\A\to{S})\subseteq(\A\to{}M)$, $\pco(S)\subseteq\pco(M)$ and
$\pnco(S)\subseteq\pnco(M)$ (see Remark \ref{jmath.0}(a)).
The proof for part (a) follows by using Remark \ref{jmath.04}(a).
The proof for part (e) is similar to the proof for part (d). We provide below
the proofs for parts (b), (c) and (d).
\begin{itemize}
\item[(b)] 
We have $\co{S_1\times{}S_2}{M_1\times{}M_2}=
\{ (x_1,x_2) \mid (x_1,x_2)\in{M_1\times{}M_2},(x_1,x_2)=\lim_i(x_1^i,x_2^i),$
\item[]\hspace*{5cm}$(x_1^i,x_2^i)_i \textnormal{\,\,is a Cauchy sequence
in\,} S_1\times{}S_2 \}$ \qquad[Remark \ref{jmath.04}(b)]
\item[]$=\{ (x_1,x_2) \mid x_1\in{M_1}, x_1=\lim_ix_1^i, (x_1^i)_i
\textnormal{\,\,is a Cauchy sequence in\,} S_1,$
\item[]\hspace*{5cm}$x_2\in{M_2}, x_2=\lim_ix_2^i, (x_2^i)_i
\textnormal{\,\,is a Cauchy sequence in\,} S_2\}$
\item[]$=\{ x_1 \mid x_1\in{M_1}, x_1=\lim_ix_1^i, (x_1^i)_i
\textnormal{\,\,is a Cauchy sequence in\,} S_1\} \times$
\item[]\quad$\{ x_2 \mid x_2\in{M_2}, x_2=\lim_ix_2^i, (x_2^i)_i
\textnormal{\,\,is a Cauchy sequence in\,} S_2\}$
\item[]$=\co{S_1}{M_1}\times\co{S_2}{M_2}$.
\item[(c)]
Let $f\in{}\A\to\co{S}{M}$ (the space $\co{S}{M}$ is complete, by Remark 
\ref{jmath.03}). We define a Cauchy sequence $(f_i)_i$ in $\A\to{}S$ 
(\mbox{$f_i\in{}\A\to{}S$}, for all $i\in\Nset$) as follows: for each 
$a\in{}A$, since $f(a)\in\co{S}{M}$, we consider a Cauchy sequence 
$(x_i^a)_i$ in $S$ ($x_i^a\in{}S$ for all $i\in\Nset$) such that 
$\lim_ix_i^a=f(a)$. Without loss of generality, we may assume that 
\mbox{$\forall{}i\geq{}n\ [d(x_i^a,f(a))\leq2^{-n}]$} for any $n\in\Nset$.

For all \mbox{$i\in\Nset$}, we define $f_i\in{}\A\to{}S$ by $f_i(a)=x_i^a$,
for each $a\in{}A$. One can check that $(f_i)_i$ is a Cauchy sequence in
$\A\to{}S$ and $\lim_if_i=f$. By remarks \ref{jmath.03} and \ref{jmath.0}(a),
\mbox{$\A\to\co{S}{M}\lhd{}\A\to{}M$}, and so
$f\in{}\A\to{}M$. Therefore, \mbox{$f\in\co{\A\to{}S}{\A\to{}M}$}. Since $f$
was arbitrarily selected, we obtain
$\A\to\co{S}{M}\subseteq\co{\A\to{}S}{\A\to{}M}$.\\ 
\hspace*{0.5cm}Next, let
$f\in\co{\A\to{}S}{\A\to{}M}$. Then $f=\lim_if_i$, where $(f_i)_i$ is a
Cauchy sequence in \mbox{$\A\to{}S$}. It is easy to verify that, since
$(f_i)_i$ is a Cauchy sequence in $\A\to{}S$, $(f_i(a))_i$ is a Cauchy
sequence in $S$ for each $a\in{}A$.
Therefore, since $(f_i(a))_i$ is a Cauchy sequence in $S$, one can check that
$\lim_if_i(a)=f(a)\in{}M$ for each $a\in{}A$. Hence,
\mbox{$f\!\in\!{}\A\to\co{S}{M}$}, which means that we have
$\co{\A\to{}S}{\A\to{}M}\subseteq{}\A\to\co{S}{M}$. We conclude that
$\A\to\co{S}{M}=\co{\A\to{}S}{\A\to{}M}$. 
\item[(d)]
First, we observe that $\emptyset\in\pco(\co{S}{M})$, and also
$\emptyset\in\co{\pco(S)}{\pco(M)}$. \\ 
\hspace*{0.5cm}Next, let $X\in\co{\pco(S)}{\pco(M)}$, $X\neq\emptyset$. Since
$X\in\co{\pco(S)}{\pco(M)}$, then $X\in\pco(M)$ and $X=\lim_iX_i$, where $(X_i)_i$
is a Cauchy sequence with \mbox{$X_i\in\pco(S)$} for all $i\in\Nset$. By
Theorem~\ref{theory.metric.4} (assuming that $X_i\neq\emptyset$ for almost
all $i$, say for $i\geq{}n$), we have
\begin{itemize}
\item[]$X=\lim_iX_i=\{ \lim_{i\geq{}n}x_i \mid \forall{}i\geq{}n: x_i\in{}X_i$,
$(x_i)_{i=n}^{\infty}$ is a Cauchy sequence in $S\}$
\item[]$=\{ x \mid x\in{}M, x=\lim_{i\geq{}n}x_i, \forall{}i\geq{}n: x_i\in{}X_i$, 
$(x_i)_{i=n}^{\infty}$ is a Cauchy sequence in $S\}$
\item[]$\subseteq \{ x \mid x\in{}M, x=\lim_ix_i, \forall{}i\in\Nset: x_i\in{}S$, 
$(x_i)_i$ is a Cauchy sequence in $S\} = \co{S}{M}$.
\end{itemize}
Since $X$ is compact and $X\subseteq\co{S}{M}$, we have $X\in\pco(\co{S}{M})$.\\
Therefore, $\co{\pco(S)}{\pco(M)}\subseteq\pco(\co{S}{M})$. \\
\hspace*{0.5cm}
The proof that $\pco(co(S|M))\subseteq co(\pco(S)|\pco(M))$
follows by using Theorem~\ref{theory.metric.5}.
\end{itemize}
\end{proof}

\subsection{Denotable continuations}

The completeness condition of the weak abstraction criterion presented in 
this paper uses a notion of {\em{denotable continuation}}. The class of 
denotable continuations represents an invariant of the computation, and its 
definition relies on a construction that employs a compliance notion in 
function spaces (presented in Definition \ref{jmath.1}). The class of 
denotable continuations is introduced formally in Definition \ref{jmath.14}.

\begin{defn} \label{jmath.1}
Let $(M_1,d_1)$ and $(M_2,d_2)$ be metric spaces. 
Let $S_1$ and $S_2$ be nonempty subsets of $M_1$~and~$M_2$, respectively, 
$S_1\subseteq{}M_1$, and $S_2\subseteq{}M_2$.  
We define the metric space 
\mbox{$(\inv{M_1}{S_1}{M_2}{S_2},\dc)$} by:\\
\centerline{$\inv{M_1}{S_1}{M_2}{S_2}=\{ \f \mid \f\in{}M_1\fun{1}M_2, (\forall{}x\in{}S_1\,:\,\f(x)\in{}S_2)\, \}$
\quad\qquad$\dc=\frestr{\df}{\inv{M_1}{S_1}{M_2}{S_2}}$,}\\
where $\df$ is the standard metric defined on $M_1\fun{1}M_2$ \cite{bv96},%
\footnote{The metric defined on $M_1\fun{1}M_2$ is also presented in \cite{ct12} (Definition 2.7).}
and $\dc$ is the restriction of $\df$ to $\inv{M_1}{S_1}{M_2}{S_2}$. 
We say that $(\inv{M_1}{S_1}{M_2}{S_2},\dc)$ is an {\em{$S_1\to{}S_2$ 
compliant function space}}.
\end{defn}

Clearly, $\inv{M_1}{S_1}{M_2}{S_2}$ is a subset of $M_1\fun{1}M_2$:   
\mbox{$\inv{M_1}{S_1}{M_2}{S_2}\subseteq\,$}\mbox{$M_1\fun{1}M_2$} 
($\inv{M_1}{S_1}{M_2}{S_2}$ contains all nonexpansive functions $f\!\in\!{}M_1\!\fun{1}\!M_2$,
that in addition satisfy the property: \mbox{$\!(\forall{}x\!\in\!{}S_1\!:\!\f(x)\!\in\!{}S_2)$}).  

\begin{rem} \label{jmath.2}
As in Definition~\ref{jmath.1}, let $(M_1,d_1)$ and $(M_2,d_2)$ be metric spaces. 
Let $S_1$ and $S_2$ be nonempty subsets of $M_1$~and~$M_2$, respectively, 
$S_1\subseteq{}M_1$, $S_2\subseteq{}M_2$. 
One can establish the properties presented below. 
\begin{itemize}
\item[(a)] $(\inv{M_1}{S_1}{M_2}{S_2},\dc)$ is a subspace of $(M_1\fun{1}M_2,\df)$): $\inv{M_1}{S_1}{M_2}{S_2}\lhd{}M_1\fun{1}M_2$.  
\item[(b)] If  $(M_1,d_1)$ and $(M_2,d_2)$ are ultrametric then 
$(\inv{M_1}{S_1}{M_2}{S_2},\dc)$ is also an ultrametric space.
\item[(c)] 
The sets $S_1$ and $S_2$ can be endowed with the metrics $\frestr{d_1}{S_1}$ 
and $\frestr{d_2}{S_2}$, respectively. If the spaces $(M_2,d_2)$ and 
$(S_2,\frestr{d_2}{S_2})$ are complete then $(\inv{M_1}{S_1}{M_2}{S_2},\dc)$ 
is also a complete metric space.
\end{itemize}
\end{rem} 

\begin{rem} \label{jmath.3}
Let $M_1$ and $M_2$ be metric spaces, with subspaces $S_1$ and $S_2$ such that 
$S_1\lhd{}M_1$ and $S_2\lhd{}M_2$.  Then 
we can construct the $S_1\to{}S_2$ compliant space $(\inv{M_1}{S_1}{M_2}{S_2},\frestr{\df}{\inv{M_1}{S_1}{M_2}{S_2}})$,
and (since, by Remark \ref{jmath.03},  $\co{S_i}{M_i}\lhd{}M_i$, for $i=1,2$) we can also construct 
the \mbox{$\co{S_1}{M_1}\to\co{S_2}{M_2}$} compliant function space
$(\inv{M_1}{\co{S_1}{M_1}}{M_2}{\co{S_2}{M_2}},\frestr{\df}{\inv{M_1}{\co{S_1}{M_1}}{M_2}{\co{S_2}{M_2}}}$. 
\end{rem}

\begin{lem} \label{jmath.4}
Let $(M_1,d_1)$ and $(M_2,d_2)$ be complete metric spaces. 
Let $S_1$ and $S_2$ be nonempty subsets of $M_1$ and $M_2$ such that 
$S_1\subseteq{}M_1$ and $S_2\subseteq{}M_2$. If $\f\in\inv{M_1}{S_1}{M_2}{S_2}$, 
then $\f\in\inv{M_1}{\co{S_1}{M_1}}{M_2}{\co{S_2}{M_2}}$.
\end{lem}

\begin{cor} \label{jmath.5} 
Let $(M_1,d_1)$ and $(M_2,d_2)$ be complete metric spaces, 
and $S_1$ and $S_2$ be nonempty subsets of~$M_1$ and~$M_2$
($S_1\subseteq{}M_1$, $S_2\subseteq{}M_2$). Then we have 
\mbox{$\inv{M_1}{S_1}{M_2}{S_2}\lhd\inv{M_1}{\co{S_1}{M_1}}{M_2}{\co{S_2}{M_2}}$}. 
\end{cor}

\begin{rem} \label{jmath.7}
In the metric approach \cite{bv96}, a continuation-based denotational semantics 
\mbox{$\dsem:\Lang\to\D$} is a function which maps elements of a
language $\Lang$ to values in a domain $\D\cong\C\fun{1}\Rd$,
where $\D$ is the domain of {\em{computations}} (or {\em{denotations}}),
$\C$ is the domain of {\em{continuations}}
and $\Rd$ is a domain of final answers.
Note that $\D$, $\C$ and $\Rd$ are metric domains, i.e., complete metric spaces. 
In general, the domain of continuations $\C$ is given by an equation  of the form 
$\C=\cdots(\half\cdot\D)\cdots$, i.e., the definition of $\C$ depends on the domain $\D$.
In this paper, we consider only denotational semantics designed using 
domain equations of the form \mbox{$\D\cong\C\fun{1}\Rd$} 
that have unique solutions (up to isometry $\cong$) 
\cite{ar89}, and we focus on ultrametric domains.
\footnote{In the applications presented in this article, 
the domains $\D$, $\C$ and $\Rd$ are complete ultrametric spaces.}
\end{rem}

The {\em{semantic operators}} that are used in the definition of a denotational
semantics \mbox{$\dsem:\Lang\to\D$} are {\em{nonexpansive 
functions}} that receive as arguments and 
yield as results values of various types, including (combinations of) the domains $\D$, $\C$ and $\Rd$.  

\begin{defn} \label{jmath.8}
Let $\dsem:\Lang\to\D$ be a continuation-based denotational semantics, 
where the semantic domain \mbox{$\D\cong\C\fun{1}\Rd$} is as in Remark \ref{jmath.7}. 
We define two classes of metric domains $\Aclass$ and $\Fclass$ for $\dsem$,
class $\Aclass$ with typical element $\Arg$ and 
class $\Fclass$ with typical element $\F$: 

\centerline{$\Arg \xbnf \Md \xsep \D \xsep \C \xsep \Arg\times\Arg 
\xsep \Arg+\Arg \xsep \A\to\Arg \xsep \pco(\Arg) \xsep \pnco(\Arg)$ ,}

\centerline{$\F \xbnf \Arg\fun{1}\Arg$ .}

An element $\Arg\in\Aclass$ is an {\em{argument type}}, an element 
$\F\in\Fclass$ is an {\em{operator type}}. 

\noindent
Here $\Md$ is an arbitrary metric domain (a complete metric space) 
that does not depend on~either~$\D$~or~$\C$.%
\footnote{In particular, $\Md$ could be $\Rd$, $\Md=\Rd$ (in case $\Rd$ does not depend on~either~$\D$~or~$\C$).
In the applications presented in this paper the final domain $\Rd$ does not depend on~either~$\D$~or~$\C$. 
In general, domain $\Rd$ may depend on $\D$ (see chapter 18 of \cite{bv96}), 
in which case $\Rd$ may need to be modelled as a more complex  argument type. 
}
$A$ is an arbitrary set. The composed domains $\Arg\in\Aclass$ and $\F\in\Fclass$
are endowed with the standard metrics defined on the product space and the function space, respectively \cite{bv96}.%
\footnote{The metrics defined on composed spaces are also presented in \cite{ct12} (Definition 2.7).}
\end{defn}

\begin{rem} \label{jmath.80}
Since in Definition \ref{jmath.8}, $\Md$, $\D$ and $\C$ are complete spaces, 
any argument type $\Arg\in\Aclass$ is a metric domain (a complete metric space).
Also, any operator type $\F\!\in\!\Fclass$ is a metric domain~\cite{bv96}.\footnote{The completeness properties of composed spaces are also presented 
in \cite{ct12} (Remark 2.8).} 
\end{rem}

Note that the (restricted) function space $\A\to\Arg$ and the compact and
non-empty and compact powerdomain constructions $\pco(\Arg)$ and
$\pnco(\Arg)$) \footnote{Since in practice continuations
are finite structures, and since any finite set is compact \cite{bv96}, the
compactness requirement is satisfied naturally in most applications.} are
not needed in the approach presented in this paper, and are rarely used in
practice to specify argument types. The compact powerdomain constructions can
be used to specify nondeterministic behaviour by using operators for
nondeterministic scheduling \footnote{To give an example,
for a nature inspired formalism \cite{ct22}, it is presented a
denotational semantics that uses a nondeterministic scheduler mapping which
yields a collection of schedules, where each schedule is a pair consisting of
a denotation (computation) and a corresponding continuation.}. In this paper,
the specification of nondeterministic behaviour is given in the definition of
the semantic operators for parallel composition and nondeterministic choice
(without the need for nondeterministic schedulers).
However, the class $\Aclass$ can be extended with 
other constructions, including the (restricted) function space $\A\to\Arg$ 
(where $\A$ is an arbitrary set), and the compact and non-empty and compact 
powerdomain constructions $\pco(\Arg)$ and $\pnco(\Arg)$).

\begin{defn} \label{jmath.10}
We consider a continuation-based denotational semantics $\dsem:\Lang\to\D$, 
where the semantic domain is \mbox{$\D\cong\C\fun{1}\Rd$} as in Definition 
\ref{jmath.8}. Let $\Dt$ and $\Ct$ be subspaces of domains $\D$ and $\C$ such 
that~\mbox{$\Dt\lhd\D$} and~\mbox{$\Ct\lhd\C$}. For any argument type 
$\Arg\in\Aclass$, we define the metric space $\ainv{\Arg}{\Dt}{\Ct}$ (using 
induction on the structure of $\Arg$) by:\\
\centerline{$\ainv{\Md}{\Dt}{\Ct}=\Md$ \qquad$\ainv{\D}{\Dt}{\Ct}=\Dt$ \qquad $\ainv{\C}{\Dt}{\Ct}=\Ct$}\\
\centerline{$\ainv{(\Arg_1\times\Arg_2)}{\Dt}{\Ct}=(\ainv{\Arg_1}{\Dt}{\Ct})\times(\ainv{\Arg_2}{\Dt}{\Ct})$.}\\
\centerline{$\ainv{(\Arg_1+\Arg_2)}{\Dt}{\Ct}=(\ainv{\Arg_1}{\Dt}{\Ct})+
(\ainv{\Arg_2}{\Dt}{\Ct})$}\\
\centerline{$\ainv{(\A\to\Arg)}{\Dt}{\Ct}=\A\to({\ainv{\Arg}{\Dt}{\Ct}})$}\\
\centerline{$\ainv{(\pco(\Arg))}{\Dt}{\Ct}=\pco(\ainv{\Arg}{\Dt}{\Ct})$}\\
\centerline{$\ainv{(\pnco(\Arg))}{\Dt}{\Ct}=\pnco(\ainv{\Arg}{\Dt}{\Ct})$.}\\
We use a similar notation for operator types $\F\in\Fclass$. Namely, 
if $\F=\Arg_1\fun{1}\Arg_2$ (with $\Arg_1,\Arg_2\in\Aclass$), we define the space 
$\ainv{\F}{\Dt}{\Ct}$ by: \\
\centerline{$\ainv{\F}{\Dt}{\Ct}=\ainv{(\Arg_1\fun{1}\Arg_2)}{\Dt}{\Ct}=\inv{\Arg_1}{\ainv{\Arg_1}{\Dt}{\Ct}}{\Arg_2}{\ainv{\Arg_2}{\Dt}{\Ct}}$.}
For any \mbox{$\Arg\!\in\!\Aclass$} and \mbox{$\F\!\in\!\Fclass$} we have 
\mbox{$\ainv{\Arg}{\Dt}{\Ct}\lhd\Arg$} and \mbox{$\ainv{\F}{\Dt}{\Ct}\lhd\F$} 
(Remark \ref{jmath.100}), and we endow the spaces 
\mbox{$\ainv{\Arg}{\Dt}{\Ct}$} and $\ainv{\F}{\Dt}{\Ct}$
with the metrics \mbox{$\frestr{d_{\Arg}}{\!\ainv{\Arg}{\Dt}{\Ct}}\!$}
and \mbox{$\frestr{d_{\F}}{\!\ainv{\F}{\Dt}{\Ct}}$}, respectively.
\end{defn}

\begin{rem} \label{jmath.100}
Let $\dsem:\Lang\to\D$ with $\D\cong\C\fun{1}\Rd$
be a continuation-based denotational semantics, as in Definition \ref{jmath.8}. 
Let $\Dt$ and $\Ct$ be subspaces of domains $\D$ and $\C$, respectively, \mbox{$\Dt\lhd\D$} and \mbox{$\Ct\lhd\C$}.
Let $\Arg\in\Aclass$ and $\F\in\Fclass$.
\begin{itemize}
\item[(a)] The spaces $\ainv{\Arg}{\Dt}{\Ct}$ and $\ainv{\F}{\Dt}{\Ct}$
are well-defined, $\ainv{\Arg}{\Dt}{\Ct}\lhd\Arg$ and $\ainv{\F}{\Dt}{\Ct}\lhd\F$.
\item[(b)] Assuming that spaces $\Md$, $\D$ and $\C$ are ultrametric in 
Definition \ref{jmath.8}, 
$\ainv{\Arg}{\Dt}{\Ct}$ and $\ainv{\F}{\Dt}{\Ct}$ are also ultrametric spaces.
\item[(c)] Furthermore, if the spaces $\Dt$ and $\Ct$ are complete, 
then $\ainv{\Arg}{\Dt}{\Ct}$ and $\ainv{\F}{\Dt}{\Ct}$ are also complete metric spaces. 
\end{itemize}
\end{rem}

\begin{lem} \label{jmath.1000}
Let $\dsem:\Lang\to\D$ with $\D\cong\C\fun{1}\Rd$
be a continuation-based denotational semantics (as in Definition \ref{jmath.8}). 
Let $\Dt$ and $\Ct$ be subspaces of domains $\D$ and $\C$ such that 
\mbox{$\Dt\lhd\D$} and \mbox{$\Ct\lhd\C$}.\\
For all $\Arg\in\Aclass$, we have \ 
$\ainv{\Arg}{\co{\Dt}{\D}}{\co{\Ct}{\C}}=\co{\ainv{\Arg}{\Dt}{\Ct}}{\Arg}$. 
\end{lem}

\begin{defn} \label{jmath.12}   
We consider a continuation-based denotational semantics $\dsem:\Lang\to\D$, 
where the semantic domain is \mbox{$\D\cong\C\fun{1}\Rd$} as in Definition \ref{jmath.8}.
Let $\Dt$ and $\Ct$ be subspaces of domains $\D$ and $\C$ such that 
\mbox{$\Dt\lhd\D$} and \mbox{$\Ct\lhd\C$}. 
Let $\F$ be an operator type, \mbox{$\F\in\Fclass$}. 
Let $\f\in\F$ be an operator of type $\F$.
We say that the class of continuations $\Ct$ is {\em{invariant}} for $\Dt$ 
under the operator $\f$ \ iff \ $\f\in\ainv{\F}{\Dt}{\Ct}$.  
\end{defn}

\begin{lem} \label{jmath.13}
Let $\dsem:\Lang\to\D$ with $\D\cong\C\fun{1}\Rd$
be a continuation-based denotational semantics (as in Definition \ref{jmath.8}). 
Let $\Dt$ and $\Ct$ be subspaces of domains $\D$ and $\C$ such 
that \mbox{$\Dt\lhd\D$} and \mbox{$\Ct\lhd\C$}.
\begin{itemize}
\item[(a)]
$\ainv{\Arg}{\Dt}{\Ct}\lhd\ainv{\Arg}{\co{\Dt}{\D}}{\co{\Ct}{\C}}$ for any $\Arg\in\Aclass$.
\item[(b)]
$\ainv{\F}{\Dt}{\Ct}\lhd\ainv{\F}{\co{\Dt}{\D}}{\co{\Ct}{\C}}$ for any $\F\in\Fclass$.
\item[(c)] Let $\F\in\Fclass$ be an operator type, and let $\f\in\F$ be an operator of type $\F$. 
If $\Ct$ is invariant for $\Dt$ under the operator $\f\in\F$, then 
$\co{\Ct}{\C}$ is also invariant for $\co{\Dt}{\D}$ under operator $\f$.
\end{itemize}
\end{lem}

\begin{defn} \label{jmath.14}
Let $\dsem:\Lang\to\D$ with $\D\cong\C\fun{1}\Rd$
be a continuation-based metric denotational semantics. 
We put \mbox{$\Dd=\{ \dsem(s) \mid s\in\Lang \}$}. 
Since $\dsem(s)\in\D$ for any $s\in\Lang$, 
$\Dd\lhd\D$ (we endow $\Dd$ with $\frestr{d_{\D}}{\Dd}$). 
Let \mbox{$f_1\in\F_1,\cdots,f_n\in\F_n$} be all operators 
used in the definition of the denotational mapping~$\dsem$.
If $\Cd$ is a subspace of $\C$, namely $\Cd\lhd\C$,  
we say that $\Cd$ is a {\em{class of denotable continuations}} for~$\dsem$ 
iff $\Cd$ is invariant for~$\Dd$ under all operators $f_i$ ($i=1,\ldots,n$)   
used in defining $\dsem$. 
If~$\Cd$ is a class of denotable continuations for $\dsem$, the metric domain 
$\co{\Cd}{\C}$ is called a {\em{domain of denotable continuations}} for~$\dsem$. 
\end{defn}

\begin{rem} \label{jmath.140}
Let $\dsem:\Lang\to\D$ with $\D\cong\C\fun{1}\Rd$
be a continuation-based metric denotational semantics, 
and $\Dd=\{ \dsem(s) \mid s\in\Lang \}$ as in Definition \ref{jmath.14}.  
Let \mbox{$f_1\in\F_1,\cdots,f_n\in\F_n$} be all operators 
used in the definition of the denotational mapping $\dsem$.
If $\Cd$ is a class of denotable continuations for~$\dsem$ and $\co{\Cd}{\C}$
is the corresponding metric domain of denotable continuations, by Lemma~\ref{jmath.13}(c),  
$\co{\Cd}{\C}$ is invariant for $\co{\Dd}{\D}$ under all operators $f_i$ ($i=1,\ldots,n$)   
used in the definition of $\dsem$.  
\end{rem}

\subsection{Weak abstractness criterion} \label{theory.weakabs}

If we compare the classic {\em{full abstractness}} criterion \cite{mil77} 
with the {\em{weak abstractness}} criterion \cite{ct17} employed in this paper, 
we emphasize that the {\em{correctness}} condition of the two criteria coincides, 
but the weak abstractness criterion relies on a weaker completeness condition 
called {\em{weak completeness}}, a condition that should be verified {\it only
for denotable continuations}.
While the classic full abstractness condition cannot be established 
in continuation semantics \cite{car94,ct17},
the abstractness of a continuation-based denotational model can be 
investigated based on the weak abstractness criterion.
The terminology used here to present the abstraction criteria 
(comprising also the notion of a {\em{syntactic context}}) is taken from \cite{bv96}. 
We recall the {\em{completeness}} condition of the full abstractness criterion 
for such a continuation-based model\footnote{The {\em{correctness}} condition of 
the two criteria coincides, and is presented in Definition \ref{weakabs.def}(a).}. 
For this, we consider a language $\Lang$, a continuation-based denotational 
semantics $\dsem:\Lang\to\D$ (where the domain $\D\cong\C\fun{1}\Rd$ is as in 
Remark \ref{jmath.7} and $\C$ is the domain of continuations), 
and an operational semantics $\osem:\Lang\to\Od$. 
If $\scontext$ ranges over a set of {\em{syntactic contexts}} for~$\Lang$, 
$\dsem$ is {\em{complete}} with respect to $\osem$ when

\centerline{$\forall{x_1,x_2\in\Lang}\,[(\exists{\gamma}\in\C\,[\dsem(x_1)(\gamma)\neq\dsem(x_2)(\gamma)])
\Rightarrow(\exists{\scontext}\,[\osem(\scontext(x_1))\neq\osem(\scontext(x_2))])]$.}

\noindent
When the domain of continuations $\C$ contains elements which do not correspond to 
language elements, this completeness condition may not hold \cite{car94,ct17}.
Definition \ref{weakabs.def} presents the weak abstractness criterion which 
comprises a weaker completeness condition. 
In Definition \ref{weakabs.def}  and Lemma \ref{weakabs.lem}, 
we assume that $(x\in)\Lang$ is a language,  
\mbox{$\dsem:\Lang\to\D$} is a continuation-based denotational semantics 
where the denotational domain $\D\cong\C\fun{1}\Rd$ is as in Remark \ref{jmath.7},  
$(\gamma\in)\C$ is the domain of continuations, 
$\osem:\Lang\to\Od$ is an operational semantics for $\Lang$, and  
$\scontext$ is a typical element of the set of {\em{syntactic contexts}} for $\Lang$. 

\begin{defn} \label{weakabs.def} (Weak abstractness for continuation semantics)
\begin{itemize} 
\item[(a)] $\dsem$ is {\em{correct}} with respect to $\mc{O}$ iff\quad
\mbox{$\forall{x_1,x_2\in\Lang}[\dsem(x_1)=\dsem(x_2)\Rightarrow\forall{\scontext}[\osem(\scontext(x_1))=\osem(\scontext(x_2))]]$.}
\item[(b)] If $\Cd$ is a class of denotable continuations for $\dsem$ and 
$\Cdc=\co{\Cd}{\C}$ is the corresponding domain of denotable continuations for $\dsem$,  
then we say that $\dsem$ is {\em{weakly complete w.r.t $\osem$ and $\Cdc$}}~iff
\item[] \qquad $\forall{x_1,x_2\in\Lang}\,[(\exists{\gamma}\in\Cdc\,[\dsem(x_1)\gamma\neq\dsem(x_2)\gamma])
\Rightarrow (\exists{\scontext}\,[\osem(\scontext(x_1))\neq\osem(\scontext(x_2))])]$. 
\item[] We say that $\dsem$ is {\em{weakly complete with respect to $\osem$}} iff
there exists a class of denotable continuations~$\Cd$ 
such that $\dsem$ is {\em{weakly complete}} with respect to $\osem$ and $\Cdc$, where 
$\Cdc=\co{\Cd}{\C}$ is the corresponding domain of denotable continuations.
\item[(c)] $\dsem$ is {\em{weakly abstract}} with respect to $\osem$ 
\ iff \ $\dsem$ is correct and weakly complete with respect to~$\osem$.
\end{itemize}
\end{defn}

\begin{lem} \label{weakabs.lem}
Let \mbox{$\dsem:\Lang\to\D$} be a continuation-based
denotational semantics, where the domain $\D$ is given by \mbox{$\D\cong\C\fun{1}\Rd$} and $\C$ is the domain of continuations
(as in Definition \ref{weakabs.def}).
If $\Cd$ is a class of denotable continuations for $\dsem$ and $\Cdc=\co{\Cd}{\C}$ 
is the corresponding domain of denotable continuations for~$\dsem$, then 
$\dsem$ is {\em{weakly complete}} with respect to $\osem$ and $\Cdc$ iff
\begin{equation}
\forall{x_1,x_2\in\Lang}\,[(\exists{\gamma}\in\Cd\,[\dsem(x_1)\gamma\neq\dsem(x_2)\gamma])
\Rightarrow
(\exists{\scontext}\,[\osem(\scontext(x_1))\neq\osem(\scontext(x_2))])]. \label{weakc}
\end{equation}
Therefore, if there exists a class of denotable continuations $\Cd$ for $\dsem$ such that condition \eqref{weakc}
is satisfied, then $\dsem$ is {\em{weakly complete}} with respect to $\osem$. \quad
{\rm The proof of Lemma \ref{weakabs.lem} is provided in \cite{ct17}.}
\end{lem}

Since Remark \ref{jmath.140} and Lemma \ref{weakabs.lem} automatically 
extend the above properties to the entire domain of denotable continuations, it 
is enough to verify the invariance and completeness properties required by the 
weak abstraction criterion for the class of denotable continuations.

\subsection{Finite bags and the structure of continuations} \label{theory.bags}

In this paper, the structure of continuations is defined based on a 
construction for finite bags $\bag{\cdot}$, which in turn is defined based on 
a set of identifiers $\Id$. The symbols $;$ , $\parallel$ and $\backslash$ 
occurring in an identifier $\alpha\in\Id$ are used to describe the semantics 
of sequential composition, parallel composition and restriction operators, 
respectively.

\begin{defn} \label{id.1} (Identifiers)
In the sequel of the paper we use a set 
$(c\in)\Names$ of {\em{names}};
the set of names~$\!\Names\!$ is assumed to be countable, as in CCS \cite{mil89}. 
We introduce  a set of {\em{identifiers}} $(\alpha\!\in)\!\Id$ given in BNF by:\\
\centerline{$\alpha \xbnf \ahole \xsep \ars \xsep \als{\alpha} \xsep \anu{\alpha}{c} \xsep \alp{\alpha} \xsep \arp{\alpha}$.}
For substituting the hole symbol $\ahole$ occurring in an identifier $\alpha$ with $\alpha'$, 
we use the notation $\alpha(\alpha')$ given by: $\ahole(\alpha')=\alpha'$, 
$\ars(\alpha')=\ars$, $\als{\alpha}(\alpha')=\als{\alpha(\alpha')}$,
$\anu{\alpha}{c}(\alpha')=\anu{\alpha(\alpha')}{c}$,  
$\alp{\alpha}(\alpha')=\alp{\alpha(\alpha')}$, and 
$\arp{\alpha}(\alpha')=\arp{\alpha(\alpha')}$. 
\end{defn}

The symbol $\ahole$ is used as a reference to an {\em{active computation}}. 
Thus, the substitution $\alpha(\alpha')$ does not replace the symbol 
$\ahole$ when it occurs on the right-hand side of a sequential composition $\ars$.

In previous works based on the CSC (continuation semantics for concurrency) 
technique, the set of identifiers is defined as a collection of finite 
sequences $\{1,2\}^*$ or $(\Names\cup\{1,2\})^*$ endowed with a partial 
ordering relation that is used to express the structure of continuations 
\cite{ent00,ct14} and \cite{ct20}, respectively.
In this paper, we employ a new representation of continuations based on the set of 
identifiers $\Id$ introduced in Definition~\ref{id.1}. 

We define and use the predicate $\amatch:(\Id\times\Id)\to\Bool$ given by:
\begin{itemize}
\item[]$\amatch(\ahole,\alpha)=\true$, \quad $\amatch(\ars,\ars)=\true$, 
\item[]$\amatch(\als{\alpha_1},\als{\alpha_2})=\amatch(\alp{\alpha_1},\alp{\alpha_2})=\amatch(\arp{\alpha_1},\arp{\alpha_2})=\amatch(\alpha_1,\alpha_2)$, 
\item[]$\amatch(\anu{\alpha_1}{c},\anu{\alpha_2}{c})=\amatch(\alpha_1,\alpha_2)$,\quad and \,\, $\amatch(\alpha_1,\alpha_2)=\false$ \,\,otherwise.
\end{itemize}

\noindent
By structural induction on $\alpha$, one can show that \ 
$\amatch(\alpha,\alpha)$, 
$\amatch(\alpha,\alpha') \wedge \amatch(\alpha',\alpha'')\Rightarrow \amatch(\alpha,\alpha'')$, 
and \ $\amatch(\alpha,\alpha') \wedge \amatch(\alpha',\alpha)\Rightarrow\alpha=\alpha'$,
for any $\alpha,\alpha',\alpha''\in\Id$.
Thus, the relation 
$\leq=\{(\alpha,\alpha') \mid \amatch(\alpha,\alpha') \} (\subseteq\Id\times\Id)$
is a {\em{partial order}}. We write $\alpha\leq\alpha'$ to express that $(\alpha,\alpha')\in\leq$.

When we have $\alpha\leq\alpha'$, the identifier $\alpha'$ can be obtained 
from $\alpha$ by a substitution $\alpha'=\alpha(\alpha'')$ for some $\alpha''\in\Id$. 
Let $(\pi\in)\Pi=\pfin(\Id)$. For any $\pi\in\Pi$ and $\alpha\in\Id$, we use the notation:

\centerline{$\pirestr{\pi}{\alpha}=\{ \alpha' \mid \alpha'\in\pi, \alpha\leq\alpha' \}$.}

We also define the operators $\glb:(\Id\times\Id)\to\Id$ and \mbox{$\ominus:(\Id\times\Id)\to(\Id\cup\{\xundef\})$} 
with $\xundef\notin\Id$ (for $\ominus$ we use the infix notation),
as well as the predicate \mbox{$\inalpha\!:\!(\Names\times\Id)\!\to\!\Bool$} by:

\begin{itemize} 
\item[]$\glb(\alpha,\alpha)=\alpha$,\quad$\glb(\als{\alpha_1},\als{\alpha_2})=\als{\glb(\alpha_1,\alpha_2)}$, \quad$\glb(\anu{\alpha_1}{c},\anu{\alpha_2}{c})=\anu{\glb(\alpha_1,\alpha_2)}{c}$, 
\item[]$\glb(\alp{\alpha_1},\alp{\alpha_2})=\alp{\glb(\alpha_1,\alpha_2)}$,\,\,\,
$\glb(\arp{\alpha_1},\arp{\alpha_2})=\arp{\glb(\alpha_1,\alpha_2)}$,\,\, and $\glb(\alpha_1,\alpha_2)=\ahole$\,\, otherwise 
\item[]$\alpha\ominus\ahole=\alpha$, \quad $\als{\alpha_1}\ominus\als{\alpha_2}=\alp{\alpha_1}\ominus\alp{\alpha_2}=\arp{\alpha_1}\ominus\arp{\alpha_2}=\alpha_1\ominus\alpha_2$,
\item[]$\anu{\alpha_1}{c}\ominus\anu{\alpha_2}{c}=\alpha_1\ominus\alpha_2$,\quad and \quad $\alpha_1\ominus\alpha_2=\xundef$\quad otherwise
\item[]$\inalpha(c,\ahole)=\false$,\quad$\inalpha(c,\ars)=\false$,\quad$\inalpha(c,\anu{\alpha}{c'})=\xif c'=c \xthen \true \xelse \inalpha(c,\alpha)$, 
\item[]$\inalpha(c,\als{\alpha})=\inalpha(c,\alp{\alpha})=\inalpha(c,\arp{\alpha})=\inalpha(c,\alpha)$.  
\end{itemize}

One can show that $\glb(\alpha_1,\alpha_2)$ is the {\em{greatest lower bound}} of 
$\alpha_1$ and $\alpha_2$ with respect to~$\leq$.\footnote{$\glb(\alpha_1,\alpha_2)
\leq\alpha_1$, $\glb(\alpha_1,\alpha_2)\leq\alpha_2$, and if $\alpha\leq\alpha_1$ 
and $\alpha\leq\alpha_2$, then $\alpha\leq\glb(\alpha_1,\alpha_2)$ for any 
$\alpha_1,\alpha_2,\alpha\in\Id$.}
For example, considering the identifiers $\alpha_1=\anu{\alp{\alpha_1'}}{c}$ and 
$\alpha_2=\anu{\arp{\alpha_2'}}{c}$, we have $\glb(\alpha_1,\alpha_2)=\anu{\ahole}{c}$.  
Clearly, $\glb(\alpha_1,\alpha_2)\leq\alpha_1$ and $\glb(\alpha_1,\alpha_2)\leq\alpha_2$. 
Moreover, one can show that, if $\alpha\leq\alpha'$ then $\alpha'\ominus\alpha\in\Id$.
For example, $\anu{\alp{\alpha_1'}}{c}\ominus\anu{\ahole}{c}=\alp{\alpha_1'}$. 
$\inalpha(c,\alpha)=\true$ if $c$ occurs restricted in $\alpha$.

We also define the predicate $\iotatwo:((\Names\times\Id)\times(\Names\times\Id))\to\Bool$ by:
\begin{itemize}
\item[]$\iotatwo(\at{c_1}{\alpha_1},\at{c_2}{\alpha_2})= \xlet \alpha=\glb(\alpha_1,\alpha_2), \alpha_1'=\ominus(\alpha_1,\alpha), \alpha_2'=\ominus(\alpha_2,\alpha)$
\item[]\hspace*{3.55cm}$\xin (\alpha_1\neq\alpha_2) \wedge (c_1=c_2) \wedge \neg(\inalpha(c_1,\alpha_1')) \wedge \neg(\inalpha(c_2,\alpha_2'))$.
\end{itemize}

We write a pair $(c,\alpha)\in\Names\times\Id$ as $\at{c}{\alpha}$ to express 
that action $c$ is executed by a process with identifier $\alpha$.
We model multiparty interactions using the binary interaction predicate 
$\iotatwo(\at{c_1}{\alpha_1},\at{c_2}{\alpha_2})$; such an interaction is 
successful if the two actions are executed by different processes 
$\alpha_1\neq\alpha_2$, they use the same interaction channel $c_1=c_2$, and 
they are not disabled by restriction operators $\anu{\cdot}{c}$.

As in \cite{ct17,ct20}, we use the construct $\bag{\cdot}$ 
given below  to model finite bags (multisets) of computations. 
Let $(x\in)\X$ be a metric domain, i.e. a complete metric space. 
Let $(\pi\in)\Pi\!=\!\pfin(\Id)$.  We use the notation \\
\centerline{$\bag{\X} \notation \Pi\times(\Id\to{\X})$.}
\noindent
Let $\theta$ ranges over $\Id\to{\X}$. An element of type $\bag{\X}$ 
is a pair $(\pi,\theta)$, with $\pi\in\Pi$ and $\theta\in\Id\to{\X}$.\\
We define mappings \mbox{$\idbag:\bag{\X}\to\Pi$}, 
$\applybag{(\cdot)}{\cdot}:(\bag{\X}\times\Id)\to\X$, and 
\mbox{$\subsbag{\cdot}{\cdot}{\cdot}:(\bag{\X}\times\Id\times\X)\to\bag{\X}$} by:

\centerline{$
\begin{array}{rcl}
\idbag{(\pi,\theta)} &  = & \pi , \\ 
\applybag{(\pi,\theta)}{\alpha}  &  = & \theta(\alpha) ,\\ 
\subsbag{(\pi,\theta)}{\alpha}{x}  &  = & (\pi\cup\{\alpha\},\var{\theta}{\alpha}{x}) .\\
\end{array}$}

\begin{rem} \label{bags.1}
When $X$ is a plain set (rather than a metric domain),
we use the same notation $\bag{X}=\Pi\times(\Id\fun{}X)$
(with operators $\idbag$, $\applybag{(\cdot)}{\cdot}$,
\mbox{$\subsbag{\cdot}{\cdot}{\cdot}$});
only in this case $\bag{X}$ is not equipped with a metric.
\end{rem} 

\begin{notatie} \label{lists}
The semantic models given in this paper are presented using a notation for 
{\em{finite tuples}} ({\em{lists}} or {\em{sequences}}) similar to the functional 
programming language Haskell (\texttt{www.haskell.org}) notation for lists; 
we use round brackets (rather than square brackets, that we use to represent 
multisets) to enclose the elements of a tuple. The elements in a list 
are separated by commas, and we use the symbol '$:$' as the cons operation.
For example, the empty tuple is written as $()$,
and if $(e\in)S$ is a set, $e_1,e_2,e_3\in{}S$,
then $(e_1,e_2,e_3)\in{}S^3$ ($S^3=S\times{}S\times{}S$)
and $(e_1,e_2,e_3)=e_1:(e_2,e_3)=e_1:e_2:(e_3)=e_1:e_2:e_3:()$.
By convention, $S^0=\{()\}$. 
This notation is also used for metric domains~\cite{bv96}). 
For $S$ a set, we denote by $S^*$ the set of all finite (possibly empty)
sequences over $S$.
\end{notatie}

\section{Continuation semantics for $\ccsn$} \label{ccsn}

We started the semantic investigation of $\ccsn$ with the language $\Lsyn$ 
given in \cite{bv96}, based on CCS" (we only omit the CCS relabelling operator 
\cite{mil89} which is not included in neither $\ccsn$ nor $\ccsnp$ \cite{lv10}).
In this section, we consider a language $\Lccsn$ which extends 
$\Lsyn$ with the {\em{joint input}} construct of $\ccsn$ and with the process 
algebra operators {\em{left merge}} $\xlmerge$, {\em{synchronization merge}} 
$\xsyn$ and {\em{left synchronization merge}} $\xlsyn$. We refer the reader 
to \cite{ct17,ct20,ent19} for further explanations regarding these auxiliary 
operators, which are essentially needed to make any element of the weakly 
abstract domain definable \cite{car94}. We use a set $(b\in)\IAct$ of 
{\em{internal actions}} which contains a distinguished {\em{silent action}} 
$\tau$, $\tau\in\IAct$. We also use the set of {\em{names}} $(c\in)\Names$ 
(see Definition \ref{id.1}) and a corresponding set of {\em{co-names}} 
\mbox{$(\ovr{c}\in)\Conames=\{ \ovr{c} \mid c\in\Names \}$}.
It is assumed that sets $\IAct$ and $\Names\cup\Conames$ are disjoint: $\IAct\cap(\Names\cup\Conames)=\emptyset$.
The approach to recursion is based on {\em{declarations}} and {\em{guarded statements}} \cite{bv96},
and we use a set $(y\in)\Y$ of {\em{procedure variables}}. Following \cite{bv96}, 
we work (without loss of generality) with a fixed declaration $D\in\Decl$, and in 
any context we refer to such a fixed declaration $D$.

\begin{defn} \label{ccsn.1}
The syntax of $\Lccsn$ is given by the following constructs: 
\begin{itemize} 
\item[(a)] Joint inputs $(j\in)\Jn$ \quad  $j \xbnf c \xsep j\join{}j$
\item[(b)] Elementary actions $(a\in)\Act$ \quad $a \xbnf b \xsep \ovr{c} \xsep j \xsep \xstop$ 
\item[(c)] Statements $(\s\in)\Stmt$, \qquad $\s \xbnf a \xsep y \xsep \xrestr{\s}{c} \xsep \s{\xseq}\s \xsep \s+\s 
\xsep \s\xmerge\s  \xsep \s\xsyn\s \xsep \s\xlmerge\s \xsep \s\xlsyn\s$
\item[(d)] Guarded statements $(g\in)\GStmt$,\qquad$g \xbnf a \xsep \xrestr{g}{c} \xsep g{\xseq}\s \xsep g+g 
\xsep g\xmerge{}g  \xsep g\xsyn{}g \xsep g\xlmerge{}x \xsep g\xlsyn{}g$
\item[(e)] Declarations \quad $(D\in)\Decl=\Y\to\GStmt$ 
\item[(f)] Programs \qquad $(\rho\in)\Lccsn=\Decl\times\Stmt$ .
\end{itemize}
\end{defn}

The class of elementary actions $(a\in)\Act$ comprises elements of the 
following types: internal actions $b\in\IAct$, output actions $\ovr{c}\in\Conames$, 
{\em{joint inputs}} $j\in\Jn$ and the action $\xstop$ which denotes {\em{deadlock}}. 
In addition, $\Lccsn$ provides operators for {\em{sequential composition}} 
($\s\xseq\s$), {\em{nondeterministic choice}} ($\s+\s$), {\em{restriction}} 
$\xrestr{\s}{c}$, {\em{parallel composition}} or {\em{merge}} ($\s\xmerge\s$), 
{\em{left merge}} ($\s\xlmerge\s$), {\em{synchronization merge}} 
($\s\xsyn\s$), and {\em{left synchronization merge}} ($\s\xlsyn\s$). 
These operators are known from the classic process algebra theories. 
For instance, the restriction operator $\xrestr{x}{c}$ is used to make 
the name $c$ private within the scope of $x$~\cite{mil89}.

\begin{rem} \label{ccsn.2}
In $\Lccsn$, a {\em{joint input}} $j$ is a construct $j=c_1\join\cdots\join{}c_m$,
where $1\leq{}m\leq{}n$.\footnote{A joint input is written as 
$\mset{c_1,\ldots,c_m}$ in $\ccsn$ \cite{lv10}. Since in this paper we use the 
notation based on square brackets $\mset{\cdots}$ to represent {\em{multisets}} 
(as in \cite{ct16}, see Section \ref{theory}),
for a joint input we use the notation $c_1\join\cdots\join{}c_m$.}
The language $\Lccsnp$ studied in Section \ref{ccsnp} 
provides a more general construct $j=l_1\join\cdots\join{}l_m$ 
called {\em{joint prefix}}, with $1\leq{}m\leq{}n$, where 
$l_1,\ldots,l_m\in\SAct$ are synchronization actions $\SAct=\Names\cup\Conames$. 
For the remainder of this work, we assume a fixed positive natural number $\nmax\in\Nset^+$, 
such that at most $\nmax+1$ concurrent components can be involved in any (multiparty) interaction.
However, note that $\nmax$ is a parameter of our formal specifications, 
and can be chosen to be arbitrarily large. For the language $\Lccsn$ we put $\nmax=n$
($n$~is the same number that occurs in the name of calculus $\ccsn$, in the name of language $\Lccsn$ 
and in the name of class $\Jn$). In $\Lccsn$ (as in $\ccsn$ \cite{lv10}) 
$m+1$ actions $c_1\join\cdots\join{}c_m$, 
$\ovr{c}_1,\ldots,\ovr{c}_m$, executed by $m+1$ concurrent processes
can synchronize and their interaction
is seen abstractly as a silent action $\tau$. 
\end{rem}

\begin{defn} \label{ccsn.3} 
In inductive reasoning, we use a complexity measure $\cx:\Stmt\to\Nset$ 
defined as in \cite{bv96}: 
\mbox{$\cx(a)=1$}, 
\mbox{$\cx(y) = 1+\cx(D(y))$},
\mbox{$\cx(\xrestr{x}{c})=1+\cx(x)$}, 
$\cx(x_1;x_2)\!=\!\cx(x_1\xlmerge{}x_2)\!=\!1+\cx(x_1)$,
and \mbox{$\cx(x_1 \xop x_2)=$} $1+\max\{\cx(x_1),\cx(x_2)\}$, for 
$\xop\in\{+,\xmerge,\xsyn,\xlsyn\}$.
\end{defn}

\begin{defn} \label{ccsn.4} (Interaction function for $\Lccsn$)
Let \mbox{$(u\in)\U=\{u \mid u\in\pfin(\Act\times\Id), \xcard{u}\leq\nmax+1\}$}
be the class of {\em{interaction sets}}. 
We write a pair $(a,\alpha)\in\Act\times\Id$ as $\at{a}{\alpha}$. 
We define the {\em{interaction function}} 
\mbox{$\iota: \U\to(\IAct\cup\{\xundef\})$} (where \mbox{$\xundef\notin\IAct$}) by:
$\iota(\{\at{b}{\alpha}\})=b$, 
$\iota(\{\at{\xstop}{\alpha}\})=\iota(\{\at{\ovr{c}}{\alpha}\})=\iota(\{\at{j}{\alpha}\})=\xundef$, 
$\iota(\{\at{\ovr{c}_1}{\alpha_1},\ldots,\at{\ovr{c}_m}{\alpha_m},\at{c_1\join\cdots\join{}c_m}{\alpha})\!=\!\!$
\mbox{$\xif \iotatwo(\at{c_1}{\alpha_1},\at{c_1}{\alpha})\wedge\cdots\wedge\iotatwo(\at{c_m}{\alpha_m},\at{c_m}{\alpha}) \!\xthen \tau \xelse \!\xundef$},
and $\iota(u)=\xundef$ otherwise. 
The actions in a set $u\in\U$ can interact \ iff \ $\iota(u)\in\IAct$. 
\end{defn}

\subsection{Final Semantic Domains} \label{final}
We employ (linear time) metric domains $(q\in)\Qd$ and $(q\in)\Qo$ defined 
by domain equations \cite{bv96}:

\centerline{$
\begin{array}{rcl}
\Qd &  \cong & \{\epsilon\}+(\IAct\times\half\cdot\Qd) ,\\
\Qo &  \cong & \{\epsilon\}+\{\delta\}+(\IAct\times\half\cdot\Qo) ,\\
\end{array}
$}

\noindent
where $\epsilon$ is the {\em{empty sequence}} and $\delta$ models {\em{deadlock}}.
The elements of $\Qd$ and $\Qo$ are finite or infinite sequences over $\IAct$,
and finite $\Qo$ sequences can be terminated with $\delta$.
Instead of $(b_1,(b_2,\ldots,(b_n,\epsilon)\ldots))$, 
$(b_1,(b_2,\ldots,(b_n,\delta)\ldots))$ and $(b_1,(b_2,\ldots))$, we write 
$b_1b_2\cdots{}b_n$, \ $b_1b_2\cdots{}b_n\delta$ \ and \ $b_1b_2\cdots$, 
respectively. 

The metric domain \mbox{$\Pd=\pnco(\Qd)$} is used as 
final domain for the denotational models presented in this paper.
For the operational semantics, we use the metric domain \mbox{$\Po=\pnco(\Qo)$}. 
The elements of $\Po$ and $\Pd$ are nonempty and compact subsets 
of $\Qd$ and $\Qo$, respectively. 
For any $b\in\IAct, q\in\Qo$ ($q\in\Qd$) and $p\in\Po$ ($p\in\Pd$), 
we use the notations $b\cdot{}q=(x,q)$ and  
$b\cdot{}p=\{b\cdot{}q \mid q\in{}p\}$. 

By $\tau^i$ and $\tau^i\cdot{}q$ we represent $\Qd$ sequences defined 
inductively as follows: $\tau^0=\epsilon, \tau^0\cdot{}q=q$ and $\tau^{i+1} = 
\tau\cdot\tau^i, \tau^{i+1}\cdot{}q = \tau\cdot(\tau^i\cdot{}q)$, for any 
$q\in\Qd$ and $i\geq0$. Also, for any $p\in\Pd$ and $i\geq0$, we put 
$\tau^i\cdot{}p=\{\tau^i\cdot{}q \mid q\in{}p\}$.

Silent steps are needed to establish the contractiveness of function $\Psi$ 
given in Definition \ref{densem.1}. Following \cite{ct20,ct22}, in the 
denotational model we use a sequence of the form $\tau^{\nmax}b$ to represent 
a successful interaction (among at most \mbox{$\nmax+1$} concurrent 
processes), namely $\nmax$ silent steps $\tau^{\nmax}$ followed by an 
internal action $b$, which describes the effect of the interaction. Deadlock 
is modelled in the denotational model by a sequence of $\nmax$ silent steps 
$\tau^{\nmax}$ (not followed by an internal action).

\begin{defn} \label{final.1}
For any $0\leq{}i\leq\nmax$, we define 
operators\, $\xnedp{i}:(\Pd\times\Pd)\fun{1}\Pd$ given by:\\
\centerline{$
p_1\xnedp{i}p_2 = \xlet \ovr{p}=\{ q \mid q\in{}p_1\cup{}p_2, q=\tau^{\nmax-i}\cdot\ovr{q}, \ovr{q}\neq\epsilon\} 
\xin \xif \ovr{p}=\emptyset \xthen \{\tau^{\nmax-i}\} \xelse \ovr{p}.$}
The operators $\xnedp{i}$ are well-defined, associative and commutative 
\cite{ct20,ct22}.
\end{defn}

\subsection{Operational semantics} \label{ccsn.osem}

The operational semantics of $\Lccsn$ is defined in the style of \cite{plo04}. 
Following \cite{bv96}, we use the term {\em{resumption}} as an operational 
counterpart of the term {\em{continuation}}.

\begin{defn} \label{ccsn.osem.1} (Resumptions and configurations) Let 
$(f\in)\SRes=\bigcup_{0\leq{}i\leq\nmax}\Stmt^i$ be the class of 
{\em{synchronous resumptions}}, where $\Stmt^i=\Stmt\times\cdots\times\Stmt$ 
($i$ times). Let us consider $(r\in)\Rc \xbnf \E \xsep \s$, where 
$\s\in\Stmt$ is a statement (Definition \ref{ccsn.1}), and $\E$ is a symbol 
denoting {\em{termination}}. Also, $(k\in)\KRes=\bag{\Rc}$ (here $\bag{\Rc}$ 
introduces a plain set, see Remark \ref{bags.1}), and 
\mbox{$(\mu\in)\Ids=\bigcup_{1\leq{}i\leq\nmax+1}\Id^i$} 
($\Id^i=\Id\times\cdots\Id$ ($i$ times). An element of the type $\Ids$ is a 
nonempty sequence of identifiers of length at most $\nmax+1$. Let 
\mbox{$\ARes=\Ids\times\KRes$} be the class of {\em{asynchronous 
resumptions}}. We write a pair $(\mu,k)\in\ARes$ \, as \, 
$\asyn{\mu}{k}$. Let $\az\in\Id$, $\az=\ahole$, and 
$\kz=(\emptyset,\xlbd{\alpha}{\E})$. We define the class of 
{\em{resumptions}} $(\varrho\in)\Res$ \, as the smallest subset of 
$(\SRes\times\U\times\ARes)$, $\Res\subseteq(\SRes\times\U\times\ARes)$ 
(where $(u\in)\U$ is the set of interaction sets presented in Definition 
\ref{ccsn.4}) satisfying the following axioms and rules:\\
\centerline{$((),\emptyset,\asyn{(\az)}{\kz})\in\Res\qquad\displaystyle{\frac{(x:f,u,\asyn{\alpha:\mu}{k})\in\Res \quad{}a\in\Act\quad\xcard{u}\leq\nmax}{(f,\{\at{a}{\alpha}\}\cup{}u,\asyn{\mu}{k})\in\Res}}$\qquad}\\
\smallskip
\smallskip
\centerline{$\displaystyle{\frac{(f,u,\asyn{\alpha:\mu}{k})\in\Res\quad{}c\in\Names}{(f,u,\asyn{\alpha\anu{\ahole}{c}:\mu}{k})\in\Res}}\qquad\displaystyle{\frac{(f,u,\asyn{\alpha:\mu}{k})\in\Res\quad\s\in\Stmt}{(f,u,\asyn{\alpha\als{\ahole}:\mu}{\subsbag{k}{\alpha\ars}{\s}})\in\Res}}\qquad$}\\
\smallskip
\smallskip
\centerline{$\displaystyle{\frac{(f,u,\asyn{\alpha:\mu}{k})\in\Res\quad\s\in\Stmt}{(f,u,\asyn{\alpha\alp{\ahole}:\mu}{\subsbag{k}{\alpha\arp{\ahole}}{\s}})\in\Res}}\qquad\displaystyle{\frac{(f,u,\asyn{\alpha:\mu}{k})\in\Res\quad\s\in\Stmt\quad(\xlen(f)+\xcard{u})<\nmax}{(x:f,u,\asyn{\alpha\alp{\ahole}:\alpha\arp{\ahole}:\mu}{k})\in\Res}}$}\\
\noindent
where $\xlen(f)$ is the length of sequence $f\in\SRes$, and $\xcard{u}$ is 
the cardinal number of the $u$.

We also define the class of {\em{configurations}} $(t\in)\Conf$ by
$\Conf=(\Stmt\times\Res)\cup\Rc$. 
\end{defn}

\begin{rem} \label{ccsn.osem.2}
If we endow the set $\Nset\times\Nset$ with the {\em{lexicographic ordering}} denoted by $\prec$, 
we can define the complexity measure \mbox{$\cres:\Res\to(\Nset\times\Nset)$} by 
$\cres(f,u,\asyn{\mu}{k})=(\xcard{u},\cb(\mu))$, where $\xcard{u}$ is the cardinal number of $u$,
and for $\mu=(\alpha_1,\ldots,\alpha_m)\in\Ids$ the mapping $\cb(\mu)$ is given by $\cb(\mu)=\sum_{1\leq{}i\leq{}m}\ca(\alpha_i)$;
here, $\ca(\alpha)$ is the size of the term $\alpha$ (i.e., the number of nodes in the abstract syntax tree of $\alpha$, $\ca:\Id\to\Nset$).
One can verify that for any rule $\displaystyle{\frac{\varrho}{\varrho'}}$ presented in Definition \ref{ccsn.osem.1} ,
we have $\cres(\varrho)\prec\cres(\varrho')$. Thus, any derivation tree proving that $\varrho\in\Res$ is finite. 
\end{rem}

Before introducing the transition relation for $\Lccsn$, we present a mapping 
$\ks:(\Id\times\KRes)\to\Rc$ that is used to transform an element $k$ of type 
$\KRes$ into a value of type $\Rc$:
\begin{itemize}
\item[]$\ks(\alpha,k)=\xlet \pi=\idbag(k) \xin$
\item[]\qquad$\xif \pirestr{\pi}{\alpha}=\emptyset \xthen \E \,\,\xelse \xif \pirestr{\pi}{\alpha}=\{\alpha\} \xthen \applybag{k}{\alpha}$
\item[]\qquad$\xelse \xif \pirestrn{\pi}{\alpha}=\{c\} \xthen \ks(\alpha\anu{\ahole}{c},k)\restrr{c}$
\item[]\qquad$\xelse \xif \alpha\ars\in\pi \xthen (\ks(\alpha\als{\ahole},k))\seqr(\applybag{k}{\alpha\ars})
\xelse  (\ks(\alpha\alp{\ahole},k))\parr(\ks(\alpha\arp{\ahole},k)),$
\end{itemize}

\noindent
where the operators $\restrr:(\Rc\times\Names)\to\Rc$, $\seqr,\parr:(\Rc\times\Rc)\to\Rc$ are given by:
$\E\restrr{}c=\E$, $\s\restrr{}c=\xrestr{\s}{c}$, 
$\E\seqr\E=\E$, $\E\seqr\s=\s$, $\s\seqr\E=\s$, $\s_1\seqr\s_2=\s_1\xseq\s_2$,
$\E\parr\E=\E$, $\E\parr\s\!=\!\s$, $\s\parr\E\!=\!\s$, \mbox{$\s_1\parr\s_2\!=\!\s_1\xmerge\s_2$}. 
For any $\pi\in\Pi$ and $\alpha\in\Id$ we use the notation 
$\pirestrn{\pi}{\alpha}$ given by: 
\begin{itemize}
\item[]$\pirestrn{\pi}{\alpha}=\{ c \mid \alpha'\in\pi, \amatchn(\alpha,\alpha')=c \in \Names \}$
\item[]$\amatchn(\ahole,\anu{\alpha_2}{c})=c$,\quad$\amatchn(\anu{\alpha_1}{c},\anu{\alpha_2}{c})=\amatchn(\alpha_1,\alpha_2)$, 
\item[]$\amatchn(\als{\alpha_1},\als{\alpha_2})=\amatchn(\alp{\alpha_1},\alp{\alpha_2})=\amatchn(\arp{\alpha_1},\arp{\alpha_2})=\amatchn(\alpha_1,\alpha_2)$, 
\item[]and \quad $\amatch(\alpha_1,\alpha_2)=\xundef$\quad otherwise.
\end{itemize}

\noindent
The type of mapping $\amatchn$ is 
$\amatchn:(\Id\times\Id)\to(\Names\cup\{\xundef\})$, with $\xundef\notin\Names$.
\smallskip

The transition relation $\to$ for language $\Lccsn$ is presented below
by using the notation \mbox{$\xstep{t}{b}{t'}$} to expresses that $(t,b,t')\in\to$.
Like in \cite{bv96}, in Definition \ref{ccsn.osem.3} we write $t_1\xinfer{}t_2$ as 
an abbreviation for $\displaystyle{\frac{\xstep{t_2}{b}{t'}}{\xstep{t_1}{b}{t'}}}$.

\begin{defn} \label{ccsn.osem.3}
The transition relation $\to$
for $\Lccsn$ is the smallest subset of 
\mbox{$\Conf\times\IAct\times\Conf$}
satisfying the rules given below. 
\begin{itemize}
\item[(A0)] $\xstep{(a,((),u,\asyn{(\alpha)}{k}))}{b}{r}$\hspace*{4.4cm}$\xif \xcard{u}\leq\nmax, \iota(\{\at{a}{\alpha}\}\cup{}u)=b, \ks(\ahole,k)=r$
\item[(R1)] $(a,(\s:f,u,\asyn{\alpha:\mu}{k}))\xinfer(\s,(f,\{\at{a}{\alpha}\}\cup{}u,\asyn{\mu}{k}))$\quad$\xif \xcard{u}\leq\nmax$
\item[(R2)] $(y,(f,u,\asyn{\alpha:\mu}{k}))\xinfer(D(y),(f,u,\asyn{\alpha:\mu}{k}))$
\item[(R3)] $(\xrestr{\s}{c},(f,u,\asyn{\alpha:\mu}{k}))\xinfer(\s,(f,u,\asyn{\alpha\anu{\ahole}{c}:\mu}{k}))$
\item[(R4)] $(\s_1\xseq\s_2,(f,u,\asyn{\alpha:\mu}{k}))\xinfer(\s_1,(f,u,\asyn{\alpha\als{\ahole}:\mu}{\subsbag{k}{\alpha\ars}{\s_2}}))$
\item[(R5)] $(\s_1+\s_2,(f,u,\asyn{\alpha:\mu}{k}))\xinfer(\s_1,(f,u,\asyn{\alpha:\mu}{k}))$
\item[(R6)] $(\s_1+\s_2,(f,u,\asyn{\alpha:\mu}{k}))\xinfer(\s_2,(f,u,\asyn{\alpha:\mu}{k}))$
\item[(R7)] $(\s_1\xlmerge\s_2,(f,u,\asyn{\alpha:\mu}{k}))\xinfer(\s_1,(f,u,\asyn{\alpha\alp{\ahole}:\mu}{\subsbag{k}{\alpha\arp{\ahole}}{\s_2}}))$
\item[(R8)] $(\s_1\xlsyn\s_2,(f,u,\asyn{\alpha:\mu}{k}))\xinfer(\s_1,(\s_2:f,u,\asyn{\alpha\alp{\ahole}:\alpha\arp{\ahole}:\mu}{k}))$
\qquad$\xif (\xlen(f)+\xcard{u})<\nmax$
\item[(R9)] $(\s_1\xsyn\s_2,(f,u,\asyn{\alpha:\mu}{k}))\xinfer(\s_1,(\s_2:f,u,\asyn{\alpha\alp{\ahole}:\alpha\arp{\ahole}:\mu}{k}))$
\qquad$\xif (\xlen(f)+\xcard{u})<\nmax$
\item[(R10)] $(\s_1\xsyn\s_2,(f,u,\asyn{\alpha:\mu}{k}))\xinfer(\s_2,(\s_1:f,u,\asyn{\alpha\alp{\ahole}:\alpha\arp{\ahole}:\mu}{k}))$
\qquad$\xif (\xlen(f)+\xcard{u})<\nmax$
\item[(R11)] $(\s_1\xmerge\s_2,(f,u,\asyn{\alpha:\mu}{k}))\xinfer(\s_1,(f,u,\asyn{\alpha\alp{\ahole}:\mu}{\subsbag{k}{\alpha\arp{\ahole}}{\s_2}}))$
\item[(R12)] $(\s_1\xmerge\s_2,(f,u,\asyn{\alpha:\mu}{k}))\xinfer(\s_2,(f,u,\asyn{\alpha\alp{\ahole}:\mu}{\subsbag{k}{\alpha\arp{\ahole}}{\s_1}}))$
\item[(R13)] $(\s_1\xmerge\s_2,(f,u,\asyn{\alpha:\mu}{k}))\xinfer(\s_1,(\s_2:f,u,\asyn{\alpha\alp{\ahole}:\alpha\arp{\ahole}:\mu}{k}))$
\qquad$\xif (\xlen(f)+\xcard{u})<\nmax$
\item[(R14)] $(\s_1\xmerge\s_2,(f,u,\asyn{\alpha:\mu}{k}))\xinfer(\s_2,(\s_1:f,u,\asyn{\alpha\alp{\ahole}:\alpha\arp{\ahole}:\mu}{k}))$
\qquad$\xif (\xlen(f)+\xcard{u})<\nmax$
\item[(R15)] $\s\xinfer(\s,((),\emptyset,\asyn{(\az)}{\kz}))$ .
\end{itemize}
\end{defn}

In a configuration $(\s,(f,u,\asyn{\alpha:\mu}{k}))$, $\alpha$ is the 
identifier of the {\em{active computation}} $\s$, and the elements contained 
in $\mu$ are identifiers of the computations contained in the synchronous 
resumption $f$. Hence,we often represent a list of type $\Ids$ by highlighting 
the first element as $\alpha:\mu$, where $\mu$ can be the empty list.
To define the behaviour of a restriction operation $\xrestr{x}{c}$ evaluated 
in a context given by identifier~$\alpha$, in rule (R3) a new (local) context 
is created indicated by the identifier $\alpha\anu{\ahole}{c}$ for the 
evaluation of statement $\s$. To model multiparty interactions, joint inputs 
and output actions are added to the interaction set $u$, and an inference 
starts according to rule (R1) searching for a set of actions that could 
possibly interact. Axiom (A0) models the transition performed when it is 
found a set of elementary statements that can interact.

\begin{defn} \label{opsem.3}
For $t\in{}\Conf$, we write $t\narrow$ to express that there are no $b,t'$ 
such that $\xstep{t}{b}{t'}$. We say that~$t$ {\em{terminates}} if $t=\E$, 
and that $t$ {\em{blocks}} if $t\narrow$ and $t$ does not terminate. 
\end{defn}

\begin{defn} \label{opsem.4} (Operational semantics $\den{O}{\cdot}$ for $\Lccsn$)
Let $(S\in)\Semo=\Conf\to\Po$ ($\Po$ was defined in Section \ref{final}). 
We define the higher order mapping $\Omega:\Semo\to\Semo$ by:
\item[]\qquad$\Omega(S)(t)=\left\{
\begin{array}{l}
\{\epsilon\}\hspace*{3cm}\xif\,  t \,\,\textnormal{terminates}\\
\{\delta\}\hspace*{3cm}\xif\,  t \,\,\textnormal{blocks}\\
\bigcup\{b\cdot{}S(t') \mid \xstep{t}{b}{t'} \} \hspace*{0.2cm}\xotherwise .
\end{array}
\right.$\\
We put $\osem=\fix(\Omega)$. 
We also define $\den{O}{\cdot}:\Stmt\to\Po$ \ by \ 
$\den{O}{\s}=\osem(\s,((),\emptyset,\asyn{(\az)}{\kz}))$. 
\end{defn}

\noindent
To justify Definition \ref{opsem.4}, we note that the mapping $\Omega$ is a 
contraction (it has a {\em{unique}} fixed point, according to Banach's Theorem).

\begin{exmp} \label{ccsn.opsem.ex}
Let $\s_1,\s_2,\s_3\in\Stmt$, \mbox{$\s_1=(b_1\xmerge{}b_2);\xstop$},  
\mbox{$\s_2=((\xrestr{((b_1;(c_1\join{}c_2))\xmerge\ovr{c}_1)}{c_1})\xmerge\ovr{c}_2)\xseq(b_2+b_3)$}, 
and 
\mbox{$\s_3=(\xrestr{((c_1\join{}c_2)\xmerge\ovr{c}_1)}{c_1})\xmerge{}\ovr{c}_2$}.
Considering $\nmax=2$, in all the examples presented in this paper we have at most 
$3 (=\nmax+1)$ concurrent components interacting in each computing step. 
We use the function $\den{O}{\cdot}$ to compute the operational semantics for 
each of the three $\Lccsn$ programs $\s_1, \s_2$ and $\s_3$.
One can check that: $\den{O}{\s_1}=\{b_1b_2\delta,b_2b_1\delta\}$,
$\den{O}{\s_2}=\{b_1\tau{}b_2,b_1\tau{}b_3\}$, and $\den{O}{\s_3}=\{\tau\}$.
\end{exmp}

\noindent
\paragraph*{\bf Implementation:}
The operational and denotational semantics presented in this paper are 
available at {\texttt{http://ftp.utcluj.ro/pub/users/gc/eneia/from24}}
as executable semantic interpreters implemented in Haskell. All $\Lccsn$ and 
$\Lccsnp$ programs presented in this paper (Example~\ref{ccsn.opsem.ex}, 
Example~\ref{ccsn.densem.ex} and Example \ref{ccsnp.ex}) can be tested by 
using these semantic interpreters.

\subsection{Denotational semantics} \label{ccsn.dsem}

We define a denotational semantics \mbox{$\sem{\cdot}:\Stmt\to\D$} for $\Lccsn$, 
where (domain $\Pd$ is given in Section~\ref{final}): 

\centerline{$
\begin{array}{rcl}
(\phi\in)\D & \cong & \Cont\fun{1}\Pd\\
(\gamma\in)\Cont & = & \Conts\times\U\times\Conta \hspace*{1.2cm} \textnormal{({\em{continuations}})}\\
(\varphi\in)\Conts & = & \sum_{i=0}^{\nmax}\Sem^i\hspace*{2.4cm} \textnormal{({\em{synchronous\,continuations}})}\\
\Conta & = & \Ids\times\K \hspace*{2.8cm} \textnormal{({\em{asynchronous\,continuations}})}\\
(\kappa\in)\K & = & \bag{\Den}\hspace*{1.5cm} (\phi\in)\Sem = \half\cdot\D\hspace*{1.5cm}(\phi\in)\Den = \{\phie\}+\half\cdot\D .\\
\end{array}
$}

\noindent
The domain equation is given by the isometry $\cong$ between complete metric spaces.
All basic sets ($\Id,\Pi, \U$ and $\Ids$) are equipped with the discrete 
metric (which is an ultrametric). The construction $\bag{\cdot}$ is presented 
in Section~\ref{theory.bags}. According to \cite{ar89,bv96}, this domain 
equation has a solution which is {\em{unique}} (up to isometry) and the 
solutions for $\D$ and all other domains presented above are obtained as 
complete ultrametric spaces.

We use semantic operators for restriction $\oprestr:(\D\times\Names)\to\D$, 
sequential composition \mbox{$\xseq:(\D\times\D)\to\D$}, nondeterministic 
choice $\xned:(\D\times\D)\to\D$, parallel composition (or merge) 
$\xmerge:(\D\times\D)\to\D$, left merge \mbox{$\xlmerge\!:\!(\D\times\D)\to\D$}, 
left synchronization merge $\xlsyn\!:\!(\D\times\D)\to\D$ and synchronization 
merge \mbox{$\xsyn\!:\!(\D\times\D)\to\D$}, defined with the aid of operators 
on continuations $\opcnu:(\Cont\times\Names)\to\Cont$, 
$\opcseq:(\D\times\Cont)\to\Cont$, $\opclpar:(\D\times\Cont)\to\Cont$ and 
$\opclsyn:(\D\times\Cont)\to\Cont$ as follows:

\begin{itemize}
\item[]$\xrestr{\phi}{c}=\xlbd{\gamma}{\phi(\cnu{\gamma}{c})}$,\qquad$\phi_1\xseq\phi_2=\xlbd{\gamma}{\phi_1(\cseq{\phi_2}{\gamma})}$,\qquad
$\phi_1\xlmerge\phi_2=\xlbd{\gamma}{\phi_1(\clpar{\phi_2}{\gamma})}$,
\item[]$\phi_1\xned\phi_2=\xlbd{\gamma}{\phi_1(\gamma)\xnedp{i}\phi_2(\gamma)$,\quad$\xwhere i=\xcardu(\gamma)}$\quad{(operators $\xnedp{i}$ are presented in Section \ref{final})},
\item[]$\phi_1\xlsyn\phi_2=\xif \xcardc(\gamma) \xthen \xlbd{\gamma}{\phi_1(\clsyn{\phi_2}{\gamma})} \xelse \{\tau^{\nmax-\xcard{u}}\}$,\qquad$\phi_1\xsyn\phi_2\!=\!\phi_1\xlsyn\phi_2\xned\phi_2\xlsyn\phi_1$
\item[] $\phi_1\xmerge\phi_2\!=\!\phi_1\xlmerge\phi_2\xned\phi_2\xlmerge\phi_1\xned\phi_1\xsyn\phi_2$,\qquad $\cnu{(\varphi,u,\asyn{\alpha:\mu}{\kappa})}{c}=(\varphi,u,\asyn{\alpha\anu{\ahole}{c}:\mu}{\kappa})$, 
\item[]$\cseq{\phi}{(\varphi,u,\asyn{\alpha:\mu}{\kappa})}=(\varphi,u,\asyn{\alpha\als{\ahole}:\mu}{\subsbag{\kappa}{\alpha\ars}{\phi}})$, 
\item[]$\clpar{\phi}{(\varphi,u,\asyn{\alpha:\mu}{\kappa})}=(\varphi,u,\asyn{\alpha\alp{\ahole}:\mu}{\subsbag{\kappa}{\alpha\arp{\ahole}}{\phi}})$, and
\item[]$\clsyn{\phi}{(\varphi,u,\asyn{\alpha:\mu}{\kappa})}=(\phi:\varphi,u,\asyn{\alpha\alp{\ahole}:\alpha\arp{\ahole}:\mu}{\kappa})$. 
\end{itemize}

\noindent
The mapping $\xcardu\!:\!\Cont\!\to\!\Nset$ is defined by $\xcardu(\varphi,u,\asyn{\alpha:\mu}{\kappa})\!=\!\xcard{u}$,
and predicate \mbox{$\xcardc\!:\!\Cont\to\Bool$} is given by $\xcardc(\varphi,u,\asyn{\alpha:\mu}{\kappa})=((\xlen(\varphi)+\xcard{u})<\nmax)$,
where $\xlen(\varphi)$ is the length of sequence $\varphi$. 
Since operators~$\xnedp{i}$ are associative and commutative \cite{ct20,ct22}, 
the operator $\xned$ is also associative and commutative. The mapping 
$\kd:(\Id\times\K)\to\Den$ is the semantic counterpart of function $\ks$ 
given in Section \ref{ccsn.osem}.

\begin{itemize}
\item[]$\kd(\alpha,\kappa)=\xlet \pi=\idbag(\kappa) \xin$
\item[]\qquad$\xif \pirestr{\pi}{\alpha}=\emptyset \xthen \phie \,\,\xelse \xif \pirestr{\pi}{\alpha}=\{\alpha\} \xthen \,\applybag{\kappa}{\alpha}$
\item[]\qquad$\xelse \xif \pirestrn{\pi}{\alpha}=\{c\} \xthen \,\kd(\alpha\anu{\ahole}{c},\kappa)\restrden{c}$
\item[]\qquad$\xelse \xif \alpha\ars\in\pi \xthen \,(\kd(\alpha\als{\ahole},\kappa))\seqden(\applybag{k}{\alpha\ars})
\xelse  \,(\kd(\alpha\alp{\ahole},\kappa))\parden(\kd(\alpha\arp{\ahole},\kappa))$.
\end{itemize}

\noindent
Here, the operators $\restrden:(\Den\times\Names)\to\Den$, $\seqden,\parden:(\Den\times\Den)\to\Den$ are given by:
\mbox{$\phie\restrden{}c=\phie$}, \mbox{$\phi\restrden{}c=\xrestr{\phi}{c}$}, 
$\phie\seqden\phie=\phie$, $\phie\seqden\phi=\phi\seqden\phie=\phi$, $\phi_1\seqden\phi_2=\phi_1\xseq\phi_2$,
$\phie\parden\phie=\phie$, $\phie\parden\phi=\phi\parden\phie=\phi$, \mbox{$\phi_1\parden\phi_2\!=\!\phi_1\xmerge\phi_2$} -- 
the notation $\pirestrn{\pi}{\alpha}$ is presented in Section \ref{ccsn.osem}.

\begin{defn} \label{densem.1}
(Denotational semantics $\sem{\cdot}$)
Let $\opa:\Act\to\Cont\to\Pd$ be given by:
\begin{itemize}
\item[]$\opa(a)((),u,\asyn{(\alpha)}{\kappa})\!=\!\xif \xcard{u}\!\leq\!\nmax \xthen 
(\!\xif (\iota(\{\at{a}{\alpha}\}\cup{}u)\!=\!b\in\IAct, \kd(\ahole,\kappa)\!=\!\phie) \xthen \{\tau^{\nmax-\xcard{u}}\cdot{b}\}$
\item[]\hspace*{1.7cm}$\xelse (\!\xif (\iota(\{\at{a}{\alpha}\}\cup{}u)\!=\!b\in\IAct, \kd(\ahole,\kappa)\!=\!\phi\in\D) \xthen \tau^{\nmax-\xcard{u}}\cdot{b}\cdot\phi(\gammaz)
\xelse \{\tau^{\nmax-\xcard{u}}\}))$
\item[]$\opa(a)(\phi:\varphi,u,\asyn{(\alpha)}{\kappa})=\xif \xcard{u}\!\leq\!\nmax \xthen \tau\cdot\phi(\varphi,\{\at{a}{\alpha}\}\cup{}u,\asyn{\mu}{\kappa}) 
\xelse \{\tau^{\nmax-\xcard{u}}\cdot{b}\},$
\end{itemize}
where $\gammaz=((),\emptyset,\asyn{(\az)}{\xlbd{\alpha}{\phie}})$ \ and \ 
$\az=\ahole$ (as in Definition \ref{ccsn.osem.1}).

For \mbox{$(S\in)\Semd=\Stmt\to\D$}, we define the higher-order mapping $\Psi:\Semd\to\Semd$ by\\
\quad\centerline{$
\begin{array}{rcl}
\Psi(S)(a) & = & \xlbd{\gamma}{\opa(a)(\gamma)}\\
\Psi(S)(y) & = & \Psi(S)(D(y))\\ 
\Psi(S)(\xrestr{\s}{c}) & = & \xrestr{\Psi(S)(\s)}{c}\\ 
\Psi(S)(\s_1;\s_2) & = & \Psi(S)(\s_1)\xseq{}S(\s_2)\\
\Psi(S)(\s_1+\s_2) & = & \Psi(S)(\s_1)\xned\Psi(S)(\s_2)\\
\Psi(S)(\s_1{\xmerge}\s_2) & = & (\Psi(S)(\s_1)\xlmerge{}S(\s_2))\xned (\Psi(S)(\s_2)\xlmerge{}S(\s_1))\xned\\
&& (\Psi(S)(\s_1)\xlsyn{}\Psi(S)(\s_2))\xned (\Psi(S)(\s_2)\xlsyn{}\Psi(S)(\s_1))\\
\Psi(S)(\s_1{\xsyn}\s_2)& = & 
(\Psi(S)(\s_1)\xlsyn{}\Psi(S)(\s_2))\xned(\Psi(S)(\s_2)\xlsyn{}\Psi(S)(\s_1))\\
\Psi(S)(\s_1{\xlmerge}\s_2) & = & \Psi(S)(\s_1)\xlmerge{}S(\s_2)\\
\Psi(S)(\s_1{\xlsyn}\s_2) & = & \Psi(S)(\s_1)\xlsyn{}\Psi(S)(\s_2).\\
\end{array}
$}

We consider $\dsem=\fix(\Psi)$, and define $\den{D}{\cdot}:\Stmt\to\Pd$ by
$\den{D}{\s}=\dsem(\s)(\gammaz)$. 
\end{defn}

Definition \ref{densem.1} can be easily justified by the techniques used in 
standard metric semantics \cite{bv96}, based on the observation that the 
definition of mapping $\Psi(S)(\s)$ is structured by induction on 
the complexity measure $\cx(\s)$ presented in Definition~\ref{ccsn.3}).

\begin{exmp} \label{ccsn.densem.ex}
Let $x_1,x_2,x_3\in\Stmt$ be as in Example \ref{ccsn.opsem.ex}.
Considering $\nmax=2$, one can check that:
$\den{D}{\s_1}=\{\tau^{\nmax}b_1\tau^{\nmax}b_2\tau^{\nmax},\tau^{\nmax}b_2\tau^{\nmax}b_1\tau^{\nmax}\}$,
$\den{D}{\s_2}=\{\tau^{\nmax}b_1\tau^{\nmax}\tau\tau^{\nmax}{}b_2,\tau^{\nmax}b_1\tau^{\nmax}\tau\tau^{\nmax}{}b_3\}$, and $\den{D}{\s_3}=\{\tau^{\nmax}\tau\}$.
For each $\s_i$ (\mbox{$i=1,2,3$}), we observe that the result of 
$\den{O}{\s_i}$ (given in Example \ref{ccsn.opsem.ex}) can be obtained from 
the yield of $\den{D}{x_i}$ if we omit the interspersed sequences 
$\tau^{\nmax}$ and replace a terminating sequence $\tau^{\nmax}$ with~$\delta$.
\end{exmp}

\section{Continuation semantics for $\ccsnp$} \label{ccsnp}

As explained in \cite{lv10}, the joint input construct of $\ccsn$ 
``induces a unidirectional information flow''. 
In \cite{lv10}, it is also studied a more general calculus called $\ccsnp$
which can be obtained from $\ccsn$ by replacing 
outputs and inputs with the {\em{joint prefix}} construct written
as $\mset{\alpha_1,\ldots,\alpha_m}$,
where each~$\alpha_i$ can be either an input action or an output action.
Since we use the symbol~$\alpha$ to represent identifiers, and the notation 
$\mset{\ldots}$ to represent multisets, we employ a different notation. 
In this section, we consider a language named~$\Lccsnp$ which can be obtained from 
$\Lccsn$ by replacing the output and joint input constructs with the joint prefix construct.
We denote the $\Lccsnp$ {\em{joint prefix}} construct as 
$l_1\join\cdots{}\join{}l_m$, where~$l$ is an 
element of the set of {\em{synchronization actions}} $\SAct$,
$(l\in)\SAct=\Names\cup\Conames$. $(c\in)\Names$ is the given set of {\em{names}} and 
$(\ovr{c}\in)\Conames=\{\ovr{c} \mid c\in\Names \}$ is the set of {\em{co-names}}.    
We use a mapping $\ovr{\cdot}:\SAct\to\SAct$, defined such that $\ovr{\ovr{c}}=c$. 
The syntax of $\Lccsnp$ is similar to the syntax of $\Lccsn$.
Only the classes of {\em{joint prefixes}} $(j\!\in)\Jpn$   
and {\em{elementary actions}}  $(a\!\in)\Act$ are specific to $\Lccsnp$, and 
they are defined as follows: 

\centerline{$j \xbnf l \xsep j\join{}j$\qquad\qquad\qquad\qquad$a \xbnf b \xsep j \xsep \xstop$ .}

\noindent
As in the case of language $\Lccsn$, $b$ is an element of the class of 
{\em{internal actions}} $IAct$ (which includes the distinguished element 
$\tau$), and $\xstop$ denotes {\em{deadlock}}. 	In $\Lccsnp$, the classes of 
{\em{statements}} $(\s\in)\Stmt$, {\em{guarded statements}} $(g\in)\GStmt$ 
and {\em{declarations}} $(D\in)\Decl$ remain as in Definition \ref{ccsn.1}.

In $\Lccsnp$, a joint prefix is a construct $l_1\join\cdots\join{}l_m$ with 
$1\leq{}m\leq{}n$, where $n$ is the number occurring in the name of the 
language $\Lccsnp$ and in the name of the syntactic class $\Jpn$. As in 
Section \ref{ccsnp}, we use the number $\nmax$ introduced in Remark 
\ref{ccsn.2} as a parameter of the formal specification of $\Lccsnp$. 
However, in the case of language $\Lccsnp$ we cannot simply put $\nmax=n$ (as 
we did for $\Lccsn$). The general rule is that $\nmax$ should be chosen 
sufficiently large such that at most $\nmax+1$ concurrent components are 
involved in a synchronous (multiparty) interaction in each computation step.

The flexibility provided by the technique of continuations (as a semantic tool)
can handle a variety of complex interaction mechanisms with only minor changes 
to the formal specifications. 
Based on this flexibility, the continuation semantics for $\Lccsnp$ can be 
obtained easily from the semantic specification of $\Lccsn$. Only one 
modification in the semantic models of $\Lccsn$ is necessary to obtain the 
corresponding semantic models for $\Lccsnp$. Namely, we need to provide a new 
definition for the interaction function $\iota$. The interaction function 
$\iota:\U\!\to\!(\IAct\cup\{\xundef\})$ for language $\Lccsnp$ is defined by 
using $\msync\!:\!\U\!\to\!(\IAct\cup\{\xundef\})$ in the following way:

\vspace*{0.1cm}
\centerline{$\iota(\{\at{b}{\alpha}\})=b$,\quad$\iota(\{\at{\xstop}{\alpha}\})=\iota(\{\at{j}{\alpha}\})=\xundef$\quad{and}\quad$\iota(u)=\msync(u)$\ otherwise, \ where}
\vspace*{-0.3cm}
\begin{itemize}
\item[]$\msync(u)=\xif u=\{\at{j_1}{\alpha_1},\ldots,\at{j_m}{\alpha_m}\}, j_i\in\Jpn$\quad for all \,\,$i=1,\ldots,n$
\item[]\hspace*{2cm}$\!\!\xthen \xlet w_r=\mrcv{\at{j_1}{\alpha_1}}\uplus\cdots\uplus\mrcv{\at{j_1}{\alpha_1}}, w_s=\msnd{\at{j_1}{\alpha_1}}\uplus\cdots\uplus\msnd{\at{j_1}{\alpha_1}}$
\item[]\hspace*{2.8cm}$\xin \!\xif (\xcard{w_r}\!=\!\xcard{w_s} \wedge \exists \varpi_r\!\in\!\perm(w_r), \varpi_s\!\in\!\perm(w_s) [\match(\varpi_r,\varpi_s)]) 
\xthen \tau \xelse \xundef$
\item[]\hspace*{2cm}$\!\xelse \xundef$ .
\end{itemize}
\vspace*{-0.4cm}
\noindent

Neither in $\Lccsn$ nor in $\Lccsnp$ we impose the condition that the actions 
contained in a {\em{joint input}} or a {\em{joint prefix}} are distinct 
(these constructions describe multisets of actions). In the definition of 
$\msync$, we let $w$ range over the set $\mset{\SAct\times\Id}$ of finite 
multisets of elements of type $\SAct\times\Id$, and $\uplus$ is the multiset sum 
operator (the notation for multisets is as in \cite{ct16}). Also, we let 
$\varpi$ range over the set $(\SAct\times\Id)^*$ of finite sequences over 
$\SAct\times\Id$; for the representation sequences we use lists 
(notation~\ref{lists}). We assume that $\perm$ is a function which computes 
the permutations of a multiset. { }
If the mapping $\msync(u)$ receives as argument a set 
$u=\{\at{j_1}{\alpha_1},\ldots,\at{j_m}{\alpha_m}\}$ (where $j_i\in\Jpn$),
then it splits the collection of actions contained in $u$ into two multisets 
$w_r$ and $w_s$ containing input actions and output actions, respectively.
For this purpose, it uses two mappings 
$\mrcv{\cdot}:(\Jpn\times\Id)\to\mset{\Names\times\Id}$ and 
$\msnd{\cdot}:(\Jpn\times\Id)\to\mset{\Conames\times\Id}$. 
For example, if $j=c_1\join{}c_1\join{}\ovr{c}_2\join\ovr{c}_3$ then 
$\mrcv{\at{j}{\alpha}}=\mset{\at{c_1}{\alpha},\at{c_1}{\alpha}}$ and 
$\msnd{\at{j}{\alpha}}=\mset{\at{\ovr{c}_2}{\alpha},\at{\ovr{c}_3}{\alpha}}$.
The function $\msync$ yields $\tau$ when it finds a pair of permutations of 
$w_r$ and $w_s$ that can match; it uses function $\match$, which in turn uses 
the binary interaction mapping $\iotatwo$ presented in Section~\ref{theory.bags}.

\begin{itemize}
\item[]$\mrcv{\at{c}{\alpha}}\!=\!\mset{\at{c}{\alpha}} \xif c\!\in\!\Names$, $\mrcv{\at{\ovr{c}}{\alpha}}\!=\!\mset{} \xif \ovr{c}\in\Conames$, \, and\, $\mrcv{\at{(j_1\join{}j_2)}{\alpha}}\!=\!\mrcv{\at{j_1}{\alpha}}\uplus\mrcv{\at{j_2}{\alpha}}$ 
\item[]$\msnd{\at{\ovr{c}}{\alpha}}\!=\!\mset{\at{\ovr{c}}{\alpha}} \xif \ovr{c}\!\in\!\Conames$, $\msnd{\at{c}{\alpha}}\!=\!\mset{} \xif c\in\Names$, \, and\, $\msnd{\at{(j_1\join{}j_2)}{\alpha}}\!=\!\msnd{\at{j_1}{\alpha}}\uplus\msnd{\at{j_2}{\alpha}}$ 
\item[]$\match((),())=\true$
\item[]$\match(\at{c_r}{\alpha_r}:\varpi_r,\at{c_s}{\alpha_s}:\varpi_s)= \xif \iotatwo(\at{c_r}{\alpha_r},\at{\ovr{c}_s}{\alpha_s}) \xthen \match(\varpi_r,\varpi_s) \xelse \false$
\item[]and \quad $\match(\varpi_r,\varpi_s)=\false$\quad otherwise.
\end{itemize}

Apart from this new definition of function $\iota$,
all other components of the formal specification of $\Lccsnp$ 
(including the semantic domains and all semantic operators) 
remain as in Section \ref{ccsn}, for both the operational and the 
denotational semantics. Thus, we define the operational semantics 
\mbox{$\den{O}{\cdot}:\Stmt\to\Po$} and the denotational semantics 
$\den{D}{\cdot}:\Stmt\to\Pd$ for $\Lccsnp$ as in Definitions \ref{opsem.4} 
and~\ref{densem.1}.

\begin{exmp} \label{ccsnp.ex} 
Let $\s_4\in\Stmt$, 
\mbox{$\s_4=\xrestr{((((\ovr{c}_1\join{}c_2)\xseq{}b_1)\xmerge(c_1\join{}\ovr{c}_3))}{c_1}\xmerge((\ovr{c}_2\join{}c_3)\xseq{}b_2)$}.   
This $\Lccsnp$ program is based on a $\ccsnp$ example presented in \cite{lv10}.
Considering $\nmax=2$, one can check that we have: 
$\den{O}{\s_4}=\{\tau{}b_1b_2,\tau{}b_2b_1\}$, and $\den{D}{\s_4}=
\{\tau^{\nmax}\tau\tau^{\nmax}{}b_1\tau^{\nmax}b_2,\tau^{\nmax}\tau{}\tau^{\nmax}b_2\tau^{\nmax}b_1\}$.
\end{exmp}

\section{Weak abstractness of continuation semantics} \label{weakabs}

Since the domain of CSC is not fully abstract \cite{ct17}, we study the 
abstractness of continuation semantics based on the weak abstractness 
principle presented in Section \ref{theory.weakabs}; thus, we show that the 
denotational models given in this article are weakly abstract with respect to 
their corresponding operational models. The proofs for $\Lccsn$ and $\Lccsnp$ 
are similar. Due to space limitations, we focus on the weak abstractness 
proof for $\Lccsn$, and show that the denotational semantics $\dsem:\Stmt\to\D$ 
presented in Definition \ref{densem.1} is weakly abstract 
with respect to the operational semantics $\den{O}{\cdot}:\Stmt\to\Po$ 
presented in Definition \ref{opsem.4}.

We consider the class of {\em{syntactic contexts}} for $\Lccsn$
with typical element $\scontext$ (see Definition~\ref{weakabs.scontext}).

\begin{defn} \label{weakabs.scontext} (Syntactic contexts for $\Lccsn$) 
\ \ $\scontext \xbnf \shole \xsep a \xsep y \xsep \xrestr{\scontext}{c} \xsep 
\scontext;\scontext \xsep \scontext+\scontext \xsep \scontext{\,\xmerge\,}
\scontext \xsep \scontext{\xsyn}\scontext \xsep \scontext{\xlmerge}\scontext \xsep 
\scontext{\xlsyn}\scontext$. 
\end{defn}

\noindent
For a context $\scontext$ and statement $\s$, we denote by $\scontext(\s)$ 
the result of replacing all occurrences of $\shole$ in $\scontext$ with~$\s$.

\begin{defn} \label{weakabs.def1}
Let $\semx:\Qo\to\Qd$ be the (unique) function \cite{bv96} satisfying
$\semx(\epsilon)=\epsilon$, $\semx(\delta)=\tau^{\nmax}$ 
and \mbox{$\semx(b\cdot{}q)=\tau^{\nmax}\cdot{}b\cdot{}\semx(q)$}. 
We also define $\Semx:\Po\to\Pd$ by $\Semx(p)=\{\semx(q) \mid q\in{}p\}$. 
\end{defn}

We can now relate $\den{D}{\cdot}$ and $\den{O}{\cdot}$ for $\Lccsn$ .
The proof of Proposition~\ref{weakabs.correct.2} can proceed by using 
Lemma~\ref{weakabs.lem} and the observation that $\Semx$ is an injective function. 
We omit the proofs for Proposition \ref{weakabs.correct.2} and Lemma 
\ref{weakabs.lem.1}; the reader can find similar results in~\cite{ent00}. 

\begin{lem} \label{weakabs.lem.1}
 $\Semx(\den{O}{\s})=\den{D}{\s}$, for all $\s\in\Stmt$.
\end{lem}

\begin{prop} \label{weakabs.correct.2}
The denotational semantics $\dsem$ is {\em{correct}} with respect to the 
operational semantics~$\den{O}{\cdot}$. 
\end{prop}

\begin{defn} \label{weakabs.cls}
Let $\semk{\cdot}:\Kres\to\K$ be given by $\semk{k}=
(\idbag(k),\xlbd{\alpha}{\xif k(\alpha)=
\E \xthen \phie \xelse \dsem(\xapply{k}{\alpha})})$. 
We define the mapping $\semf{\cdot}:\SRes\to\Conts$ by $\semf{()}=()$ and 
$\semf{(\s:f)}=\dsem(\s):\semf{f}$.  We also define the mapping 
$\semc{\cdot}:\Res\to\Cont$ by $\semc{(f,u,\asyn{\alpha:\mu}{k})}=
(\semf{f},u,\asyn{\alpha:\mu}{\semk{k}})$, and consider
$\Contd\!=\!\{ \semc{(f,u,\asyn{\alpha\!:\!\mu}{k})} \!\mid\! 
(f,u,\asyn{\alpha:\mu}{k})\!\in\!\Res \}$. Clearly, $\Contd$ is a subspace 
of~$\Cont$, i.e. $\Contd\!\!\lhd\Cont$.
\end{defn}

\begin{lem} \label{weakabs.cls.lem}
$\Contd$ is a class of denotable continuations for the denotational 
semantics $\dsem$ of $\Lccsn$.
\end{lem}
\begin{proof}
Let $\Dd\!=\!\{ \dsem(s) \!\mid\! s\in\Stmt \}$. We must show that $\Contd$ 
is invariant for $\Dd$ under all operators used in the definition of $\dsem$. 
We only handle operator $\opclsyn\!\!:\!\!(\D\!\times\!\Cont)\!\to\!\Cont$ 
given in Section~\ref{ccsn.dsem}. The other operators can be handled similarly. 
The semantic operators used in the definition of a denotational semantics 
are nonexpansive; this means that we can write
$\opclsyn$ as $\opclsyn\!:\!(\D\!\times\!\Cont)\!\fun{1}\!\Cont$, namely 
$\opclsyn\in\F$, where \mbox{$\F\!=\!(\D\!\times\!\Cont)\!\fun{1}\!\Cont$}.
We show now that $\opclsyn\!\!\in\!\F(\Dd,\Contd)$, i.e., 
$\opclsyn\!(\phi,\gamma)\!\in\!\Contd\!$ for any $\phi\!\in\!\Dd\!$, 
$\gamma\!\in\!\Contd\!$. Since $\phi\!\in\!\Dd$ and 
$\gamma\!\in\!\Contd\!$, $\phi\!=\!\dsem(\s)$ for some $\s\!\in\!\Stmt$,
and $\gamma\!=\!\semc{(f,u,\asyn{\alpha:\mu}{k})}$
$=(\semf{f},u,\asyn{\alpha:\mu}{\semk{k}})$ for some 
$(f,u,\asyn{\alpha:\mu}{k})\in\Res$.
Assuming $\xlen(f)+\xcard{u}<\nmax$, we get $\opclsyn(\phi,\gamma)$ = 
$\opclsyn(\dsem(\s),(\semf{f},u,\asyn{\alpha:\mu}{\semk{k}})) =
(\dsem(\s):\semf{f},u,\asyn{\alpha\alp{\ahole}:\alpha\arp{\ahole}:\mu}{\semk{k}}))$
= $(\semf{\s:f},u,\asyn{\alpha\alp{\ahole}:\alpha\arp{\ahole}:\mu}{\semk{k}})$
= $\semc{(\s:f,u,\asyn{\alpha\alp{\ahole}:\alpha\arp{\ahole}:\mu}{k})}=\gamma'$. 
By Definition \ref{ccsn.osem.1}, since $(f,u,\asyn{\alpha:\mu}{k})\in\Res$,
we also have 
$(\s:f,u,\asyn{\alpha\alp{\ahole}:\alpha\arp{\ahole}:\mu}{k})\in\Res$. 
Thus, we infer that 
$\gamma'\!=\!\semc{(\s\!:\!f,u,\asyn{\alpha\alp{\ahole}\!:\!\alpha\arp{\ahole}\!:\!\mu}{k})}\!\in\!\Contd$, as required.
\end{proof}

By Remark \ref{jmath.140}, the domain of denotable continuations 
$\co{\Contd}{\Cont}$ is invariant for $\co{\Dd}{\D}$ under all operators used 
in the definition of $\dsem$.
The proof of Lemma \ref{weakabs.lem.2} is by induction on the depth of the 
inference of \mbox{$(f,u,\asyn{\alpha:\mu}{k})\in\Res$} 
by using the rules given in Definition \ref{ccsn.osem.1}.

\begin{lem} \label{weakabs.lem.2} 
For any $\s\in\Stmt$, $(f,u,\asyn{\alpha:\mu}{k})\in\Res$ there is an $\Lccsn$ 
syntactic context $\scontext$ such that
\mbox{$\tau^{\xcard{u}}\cdot\dsem(\s)\semc{(f,u,\asyn{\alpha:\mu}{k})}=\den{D}{\scontext(\s)}=\dsem(\scontext(\s))(\gammaz)$}.
Furthermore, $\scontext$ does not depend on $\s$; it only depends on $(f,u,\asyn{\alpha:\mu}{k})$.
For any $\s,\s'\in\Stmt$, $\s\neq{}\s'$, $(f,u,\asyn{\alpha:\mu}{k})\in\Res$
there is a syntactic context~$\scontext$ such that: 
$\tau^{\xcard{u}}\cdot\dsem(\s)\semc{(f,u,\asyn{\alpha:\mu}{k})}=\den{D}{\scontext(\s)}$ and 
$\tau^{\xcard{u}}\cdot\dsem(\s')\semc{(f,u,\asyn{\alpha:\mu}{k})}=\den{D}{\scontext(\s')}$
(the same $\scontext$ in both equalities). Notice that in general,   
$\den{D}{\scontext(\s)}\neq\den{D}{\scontext(\s')}$. 
\end{lem}

\begin{thm} \label{weakabs.theorem}
The denotational semantics $\dsem$ of $\Lccsn$ is {\em{weakly abstract}}
with respect to the operational semantics $\den{O}{\cdot}$. 
\end{thm}
\vspace*{-0.2cm}
\begin{proof}
By Proposition \ref{weakabs.correct.2}, $\dsem$ is {\em{correct}} with respect 
to~$\den{O}{\cdot}$. For the weak completeness condition, we consider the class 
of denotable continuations $\Contd$ for $\dsem$ presented in Definition 
\ref{weakabs.cls}. By Lemma~\ref{weakabs.lem}, if we prove that
$\forall{\s_1,\s_2\in\Stmt}[(\exists{\gamma}\in\Contd[\dsem(\s_1)\gamma\neq\dsem(\s_2)\gamma])\Rightarrow
(\exists{\scontext}[\osem(\scontext(\s_1))\neq\osem(\scontext(\s_2))])]$,
then $\dsem$ is also {\em{weakly complete}} with respect to $\den{O}{\cdot}$.
Let $\s_1,\s_2\in{\Stmt}$ and $(f,u,\asyn{\alpha:\mu}{k})\in\Res$ be such that 
$\dsem(x_1)\semc{(f,u,\asyn{\alpha:\mu}{k})}\neq\dsem(x_2)\semc{(f,u,\asyn{\alpha:\mu}{k})}$
which implies 
$\tau^{\xcard{u}}\cdot\dsem(x_1)\semc{(f,u,\asyn{\alpha:\mu}{k})}\neq
\tau^{\xcard{u}}\cdot\dsem(x_2)\semc{(f,u,\asyn{\alpha:\mu}{k})}$. 
By Lemma \ref{weakabs.lem.2}, there is an $\Lccsn$ syntactic context $\scontext$ 
such that $\den{D}{\scontext(\s_1)}=\tau^{\xcard{u}}\cdot\dsem(x_1)\semc{(f,u,\asyn{\alpha:\mu}{k})}\neq$
$\tau^{\xcard{u}}\cdot\dsem(x_2)\semc{(f,u,\asyn{\alpha:\mu}{k})}=\den{D}{\scontext(\s_2)}$.
By using Lemma \ref{weakabs.lem.1}, we get 
$\Semx(\den{O}{\scontext(\s_1)})=\den{D}{\scontext(\s_1)}\neq$
$\den{D}{\scontext(\s_2)}=\Semx(\den{O}{\scontext(\s_2)})$
which implies $\den{O}{\scontext(\s_1)}\neq\den{O}{\scontext(\s_2)}$.   
Thus, we conclude that $\dsem$ is weakly complete, and therefore 
{\em{weakly abstract}} with respect to~$\den{O}{\cdot}$. 
\end{proof}

In a similar way, one can show that the denotational semantics of 
$\Lccsnp$ is weakly abstract with respect to the corresponding operational model. 

\section{Conclusion}

While the classic full abstractness condition cannot be established in 
continuation semantics \cite{car94}, the abstractness of a continuation-based 
denotational model can be investigated based on the weak abstractness 
criterion. Compared to the classic full abstractness criterion \cite{mil77}, 
the weak abstractness criterion used in this paper relies on a weaker 
completeness condition that should be verified only for a class of denotable 
continuations.

We provide the denotational and operational semantics defined by using 
continuations for two process calculi (based on CCS) able to express 
multiparty synchronous interactions. We work with metric semantics and with 
the continuation semantics for concurrency (a technique introduced by the 
authors to handle advanced concurrent control mechanisms). For the multiparty 
interaction mechanisms incorporated in both process calculi, we proved that 
the continuation-based denotational models are weakly abstract with respect 
to their corresponding operational models.

As future work, we intend to investigate the weak abstractness issue for 
nature-inspired approaches introduced in the area of membrane computing 
\cite{paun10,son21,ct22}.

\nocite{*}
\bibliographystyle{eptcs}
\bibliography{generic}

\end{document}